\documentclass[journal,twocolumn,10pt]{IEEEtran}
\usepackage{amssymb}
\usepackage{amsmath}
\usepackage{amsmath,bm}
\usepackage{amsthm}
\usepackage{graphicx}
\usepackage{subfigure}
\usepackage{cite}
\usepackage{enumerate}
\usepackage[ruled]{algorithm2e}
\usepackage{color, soul}
\usepackage{url}

\begin{document}

\newtheorem{definition}{\bf Definition}
\newtheorem{theorem}{\bf Theorem}
\newtheorem{lemma}{\bf Lemma}
\newtheorem{assumption}{\bf Assumption}
\newtheorem{remark}{\bf Remark}
\newtheorem{proposition}{\bf Proposition}
\newtheorem{corollary}{\bf Corollary}
\newtheorem{property}{\bf Property}
\newtheorem{problem}{\bf Problem}
\newcommand{\e}[1]{\ensuremath{\times 10^{#1}}}
\newcommand{\ud}{\mathrm{d}}

\renewcommand{\proofname}{\indent \textbf {Proof}}
% Definition of Proof Format
%\iffalse
\def\QEDclosed{\mbox{\rule[0pt]{1.3ex}{1.3ex}}}
\def\QED{\QEDclosed}
\def\proof{\indent{\emph{Proof: }}}
\def\endproof{\hspace*{\fill}~\QED\par\endtrivlist\unskip}
%\fi

%\renewcommand{\algorithmicrequire}{ \textbf{Input:}} %Use Input in the format of Algorithm
%\renewcommand{\algorithmicensure}{ \textbf{Output:}} %UseOutput in the format of Algorithm

% ==================================================Title and Author
\title{Realtime Profiling of Fine-Grained Air Quality Index Distribution using UAV Sensing}

\vspace{0.2cm}
\author{\IEEEauthorblockN{{Yuzhe Yang},
        {Zijie Zheng},~\IEEEmembership{Student Member,~IEEE,}
        {Kaigui Bian},~\IEEEmembership{Member,~IEEE,}\\
        {Lingyang Song},~\IEEEmembership{Senior Member,~IEEE},
        and {Zhu Han},~\IEEEmembership{Fellow,~IEEE}}\\
\vspace{-0.3cm}
\thanks{Y. Yang, Z. Zheng, K. Bian and L. Song are with School of Electrical Engineering and Computer Science, Peking University, Beijing, China (email: \{yuzhe.yang, zijie.zheng, bkg, lingyang.song\}@pku.edu.cn).}
\thanks{Z. Han is with Electrical and Computer Engineering Department, University of Houston, Houston, TX, USA (email: zhan2@uh.edu).}
}

\iffalse
% ==================================================The paper headers
\markboth{Journal of \LaTeX\ Class Files,~Vol.~14, No.~8, August~2015}%
{Shell \MakeLowercase{\textit{et al.}}: Bare Demo of IEEEtran.cls for IEEE Journals}
\fi

\maketitle

%\vspace{-0.5cm}
% ==================================================Abstract
\begin{abstract}
Given significant air pollution problems, air quality index~(AQI) monitoring has recently received increasing attention.
In this paper, we design a mobile AQI monitoring system boarded on unmanned-aerial-vehicles~(UAVs), called \emph{ARMS}, to efficiently build fine-grained AQI maps in realtime.
Specifically, we first propose the Gaussian plume model on basis of the neural network~(GPM-NN), to physically characterize the particle dispersion in the air.
Based on GPM-NN, we propose a battery efficient and adaptive monitoring algorithm to monitor AQI at the selected locations and construct an accurate AQI map with the sensed data.
The proposed adaptive monitoring algorithm is evaluated in two typical scenarios, a two-dimensional open space like a roadside park, and a three-dimensional space like a courtyard inside a building.
Experimental results demonstrate that our system can provide higher prediction accuracy of AQI with GPM-NN than other existing models, while greatly reducing the power consumption with the adaptive monitoring algorithm.
%Furthermore, the adaptive monitoring techniques and operation algorithms can balance the intrinsic tradeoff between accuracy and consumption based on proposed model, as greatly reduce battery consumption to more than 63\%.
\end{abstract}
%\vspace{0.5cm}


%\vspace{0.5cm}
\IEEEpeerreviewmaketitle

% ==================================================Intro
\section{Introduction}
\IEEEPARstart{I}{n} a recent report from the World Health Organization~(WHO)~\cite{who}, air pollution has become the world's largest environmental health risk, as one in eight of global deaths are caused by air pollution exposure each year.
Air pollution is caused by gaseous pollutants that are harmful to humans and ecosystem, especially concentrated in the urban areas of developing countries.
Thus, reducing air pollution would save millions of lives, and many countries have invested significant efforts on monitoring and reducing the emission of air pollutants.
% The governments put much effort on monitoring To evaluate daily severity of the air pollution,
Government agencies have defined air quality index~(AQI) to quantify the degree of air pollution. AQI is calculated based on the concentration of a number of air pollutants~(e.g., the concentration of ${\rm PM}_{2.5}$, ${\rm PM}_{10}$ particles and so on in developing countries).
% Typically, a high or low AQI value indicates polluted air or clean air.
A higher value of AQI indicates that air quality is ``heavily'' or ``seriously'' polluted, resulting in a greater proportion of the population may experience harmful health effects~\cite{air pollution mortality}.
To intuitively reflect AQI value of locations in either two-dimensional~(2D) or three-dimensional~(3D) area, AQI map is defined to offer such convenience~\cite{airmap}.

\subsection{Mobile AQI Monitoring}
AQI monitoring can be completed by sensors at governmental static observation stations, generating a AQI map in a local area~(e.g., a city~\cite{station}). However, these static sensors can only obtain a limited number of measurement samples in the observation area and may often induce high costs. For example, there are only 28 monitoring stations in Beijing. The distance between two nearby stations is typically several ten-thousand meters, and the AQI is monitored every 2 hours~\cite{monitor station}.
% Moreover, such reflect static distribution over a large time scale, and networks. These static sensor-based monitoring methods usually adopt predictions on unmeasured area.
To provide more flexible and accurate monitoring as well as reduce the cost, mobile devices, such as cell phones, cars and balloons are used to carry sensors and process real time measuring.
Crowd-sourced photos contributed by mass of cell phones can help depict the 2D AQI map in a large geographical region in Beijing~\cite{mobisys}, with a range of $4$km$\times$$4$km. Mobile nodes equipped with sensors can provide $100$m$\times$$100$m 2D on-ground concentration maps with relatively high resolution~\cite{node1,node2,node3}. Sensors carried by tethered balloons can build the height profile of AQI at a fixed observation height within $1000$m~\cite{balloon}. A mobile system with sensors equipped in cars and drones can help monitor ${\rm PM}_{2.5}$ in open 3D space~\cite{blueaer}, with 200m per measurement.

\begin{figure*}[!t]
\small
\centering
\includegraphics[width=0.85\textwidth]{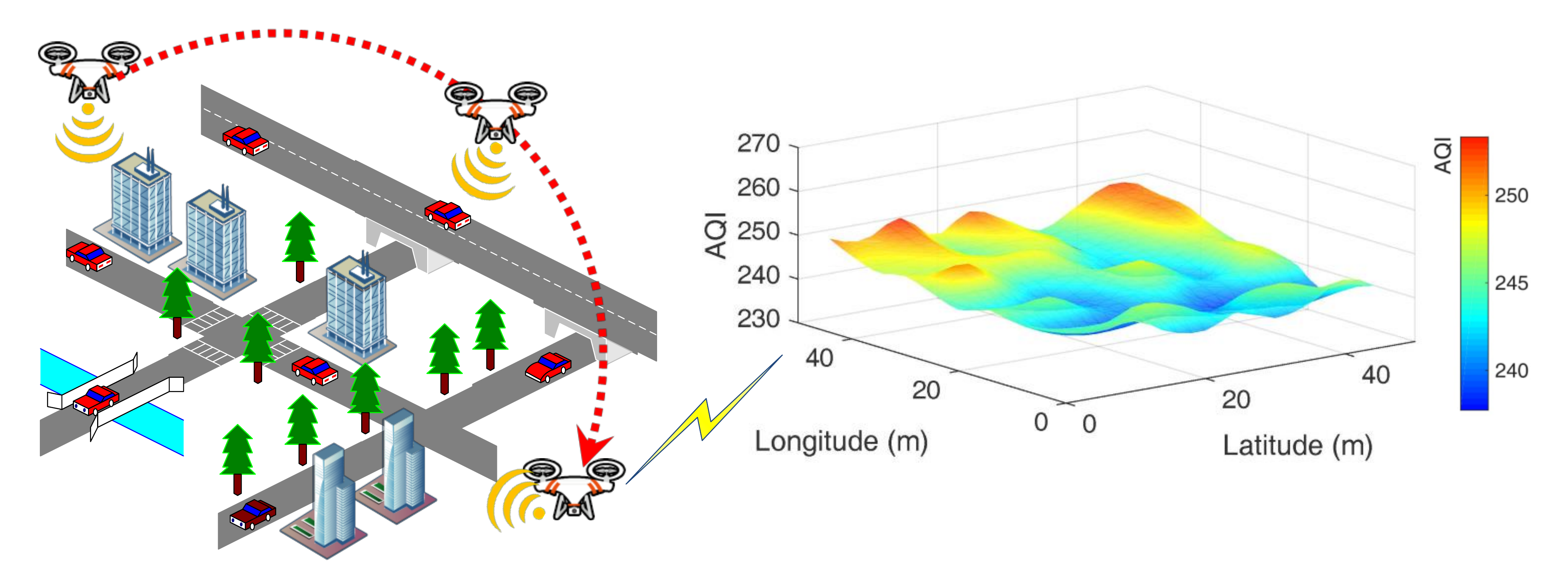}
\caption{An illustration of AQI measurement using mobile sensing over UAV.}
\vspace{-0.15cm}
\label{intro figure}
\end{figure*}

\subsection{Motivations for Realtime Fine-Grained Monitoring}
Even though current mobile sensing approaches can provide relatively accurate and real-time AQI monitoring data, they are spatially coarse-grained, since two measurements are separated by few hundreds of meters in horizontal or vertical directions in the 3D space.
However, AQI has intrinsic changes from meters to meters, and it is preferred to perform AQI monitoring in the 3D space surrounding an office building or throughout a university campus, rather than city-wide~\cite{vertical dis,height pro}. The AQI distribution in meter-sliced areas, called as fine-grained areas would be desirable for people, particularly those living in urban areas.
The fine-grained AQI map can help design the ventilation system for buildings, which for example guide teachers and students to stay away from the pollution sources on campus~\cite{whyfineaqi}.

Due to the high power consumption of mobile devices, one can only measure a limited number of locations of the entire space. % because of the battery constraint, which in turn causes coarse-grained result. To construct realtime AQI map, they must painstakingly operate many periods to generate a whole map for one time, leading to the highly reduction of survey efficiency. Thus, reasonably designing a model and choosing key measurement locations are on demand.
To avoid an exhaustive measurement, using an estimation model to approximate the value of unmeasured area has been wildly adopted.
In~\cite{uair}, the prediction model is based on a few public air quality stations and meteorological data, taxi trajectories, road networks, and Point of Interests~(POIs). However, because they estimate AQI using a feature set based on historical data, their model cannot respond in realtime to the change in pollution concentration at an hourly granularity, leading to large errors at times.
In~\cite{blueaer}, the random walk model is used for prediction by dividing the whole space into different shapes of cubes. However, the model may not reflect physical dispersion of particles~\cite{model2,model3}, and all locations are measured without considering the battery life constraint when mobile devices are used.
Mobile sensor nodes used in~\cite{node1} employ the regression model as well as graph theory to estimate the AQI value at unmeasured locations. However, they mainly focus on 2D area, and can hardly produce a 3D fine-grained map.
Neural networks~(NN) are also used for forecasting on the AQI distribution~\cite{NN1,NN2,NN3,NN4}. However, its performance in fine-gained area is not satisfied without considering the physical characteristic of real AQI distribution.

\subsection{Contributions}
To this end, in this paper we design a mobile sensing system based on unmanned-aerial-vehicles~(UAVs), called \emph{ARMS}, that can effectively catch AQI variance at meter-level and profile the corresponding fine-grained distribution.
ARMS is a realtime monitoring system that can generate current AQI map within a few minutes, compared to the previous methods with an interval of a few hours.
With ARMS, the fine-grained AQI map construction can be decomposed into two parts. First, we propose a novel AQI distribution model, named Gaussian Plume model embedding Neural Networks~(GPM-NN), that combines physical dispersion and non-linear NN structure, to do predictions of unmeasured area.
Second, we detail the adaptive monitoring algorithm as well as addressing its applications in a few typical scenarios.
By measuring only selected locations in different scenarios, GPM-NN is used to estimate AQI value at unmeasured locations and generate realtime AQI maps, which can save the battery life of mobile devices while maintaining high accuracy in AQI estimation.

% Contribution: provide effective monitoring techniques and algorithms for two common scenario
The contributions of our work are summarized as follows:
\begin{itemize}
\item[$\bullet$] The GPM-NN is highly adaptive in different fine-grained measurement scenarios, and it can provide higher accuracy in creating AQI maps than other existing models.
\item[$\bullet$] The adaptive monitoring algorithm can guide UAV to choose optimized trajectory in different scenarios based on GPM-NN. It can greatly reduce the battery consumption of ARMS, while achieving high accuracy when constructing realtime AQI maps.
\item[$\bullet$] The ARMS is the first UAV sensing system for fine-grained AQI monitoring.
\end{itemize}

The rest of this paper is organized as follows. In Section II, we briefly introduce our UAV sensing system. In Section III, we present our fine-grained AQI distribution model. The adaptive monitoring algorithm is addressed in Section IV. In Section V and Section VI, we present two typical application scenarios and performance analysis of ARMS, respectively. Finally, conclusions are drawn in Section VII.

\begin{figure*}[!t]
\small
\centering
\includegraphics[width=0.65\textwidth]{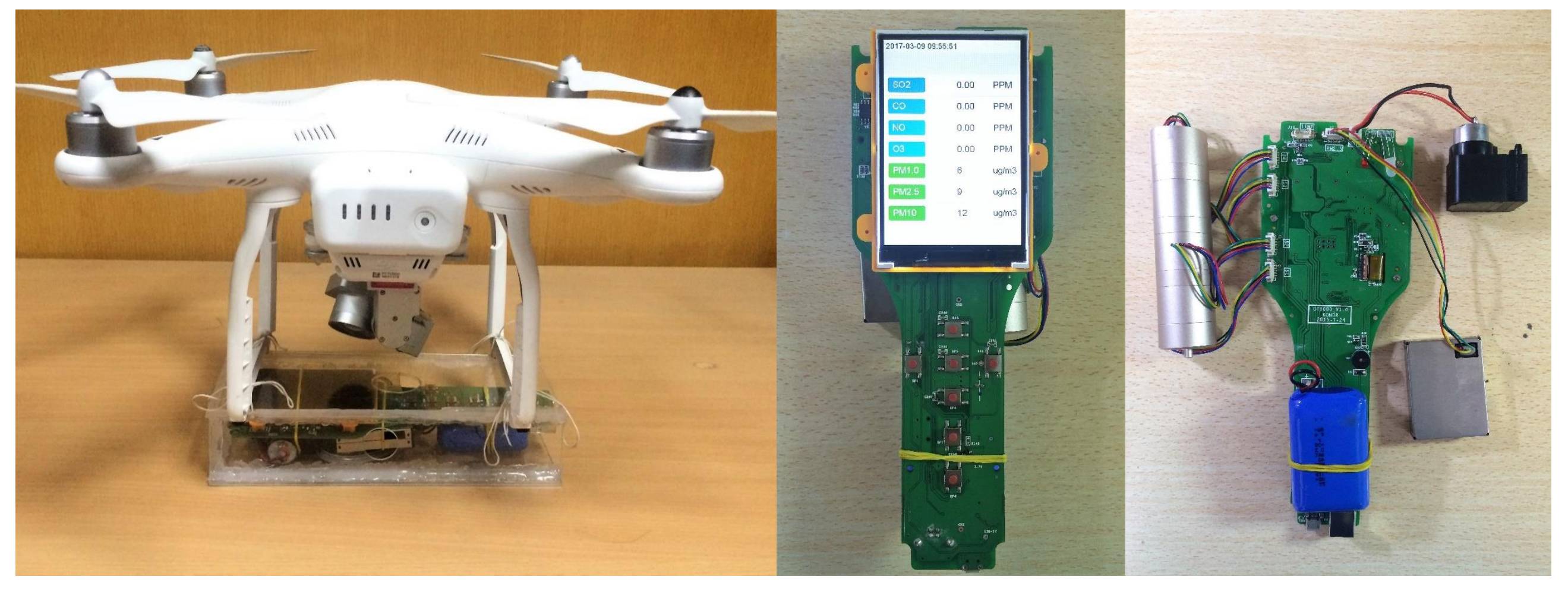}
\caption{The ARMS system, and the front and the back of the sensor board.}
%\vspace{-0.15cm}
\label{uav+sensor}
\end{figure*}

\vspace{0.3cm}
% ==================================================System Overview
\section{Preliminaries of UAV Sensing System}
In this section, first we provide a brief introduction of ARMS, and then we show how to construct a dataset using ARMS.
To confirm the reliability of the collected dataset, we compare the collected data and the official AQI measured by the nearest Beijing government's monitoring station, i.e., the Haidian station~\cite{chinaweather}.
To determine the parameters of our model, we test possible factors that may influence AQI, such as wind, locations, etc., and remove those factors that have small correlations with AQI in the fine-grained scenarios from our model.

\subsection{System Overview}
The architecture of ARMS includes an UAV and an air quality sensor boarded on the UAV, as shown in Fig.~\ref{uav+sensor}. The sensor is fixed in a plastic box with vent holes, bundled on the bottom of UAV.
The sensor uses a laser-based AQI detector~\cite{sensor}, which can provide the concentration within $\leq \pm 3\%$ monitor error for common pollutants in AQI calculation, such as ${\rm PM}_{2.5}$, ${\rm PM}_{10}$, CO, NO, ${\rm SO}_2$ and ${\rm O}_3$. The values of these pollutants are realtime recorded, with which we calculate the corresponding AQI value at measuring locations.
%In Beijing, we observe 365 days of AQI on~\cite{chinaweather} and find that ${\rm PM}_{2.5}$ is the main pollutants. Thus, in the rest of paper, when we mention AQI, the value of AQI is calculated based on the ${\rm PM}_{2.5}$ concentration.

For the UAV, we select DJI Phantom 3 Quadcopter~\cite{uav} as the mobile sensing device. The UAV can keep hosting for at most 15 minutes due to the battery constraint, which restricts the longest continuous duration within one measurement. The GPS sensor on the UAV can provide the real-time 3D position.
During one measurement, the UAV is programmed with a trajectory, including all locations that need to be measured. Following this trajectory, UAV hovers for 10 seconds to collect sufficient data to derive the AQI value at each stop, before moving to the next one.

During one monitoring process, ARMS measures all target locations and records the corresponding AQI values. After the measuring process is completed, the data is then sent to the offline PC and put into the GPM-NN model to construct the realtime AQI map. Thus, the map construction process is offline.

\subsection{Dataset Description}
Data collected by ARMS are then arranged as a dataset~\footnote{Dataset can be found at \url{https://github.com/YyzHarry/AQI_Dataset}.}.
As shown in Fig.~\ref{intro figure}, we have conducted a measurement study in both typical 2D and 3D scenarios~(i.e., a roadside park and the courtyard of an office building in Peking University), respectively, from Feb. 11 to Jul. 1, 2017, for more than 100 days to collect sufficient data.

% To our best effort of searching, existing datasets of AQI are mostly coarse-grained and nation-wide or city-wide~\cite{dataset1,dataset2,dataset3}. To address fine-grained scenarios,
% Our dataset contains extensive AQI value measured by Arms for both 2D and 3D fine-grained scenarios, respectively. The dataset now contains 60 days measuring data.
In the dataset, each $.txt$ file includes one complete measurement over a day in one typical scenario.
In each $.txt$ file, each sample has four parameters, 3D coordinates~($x,y,z$) and an AQI value. Each value represents the measured AQI, while its coordinates in the matrix reflect the position in different scenarios. In the 2D scenario, we assume $z=0$, while measuring at an interval of 5m in $x$ and $y$ directions. In the 3D scenario, every row presents fixed position in $xy$ plane, while every column represents the height at an interval of 5m in $z$ direction.

\begin{figure}[!t]
\small
\centering
\includegraphics[width=3.3in]{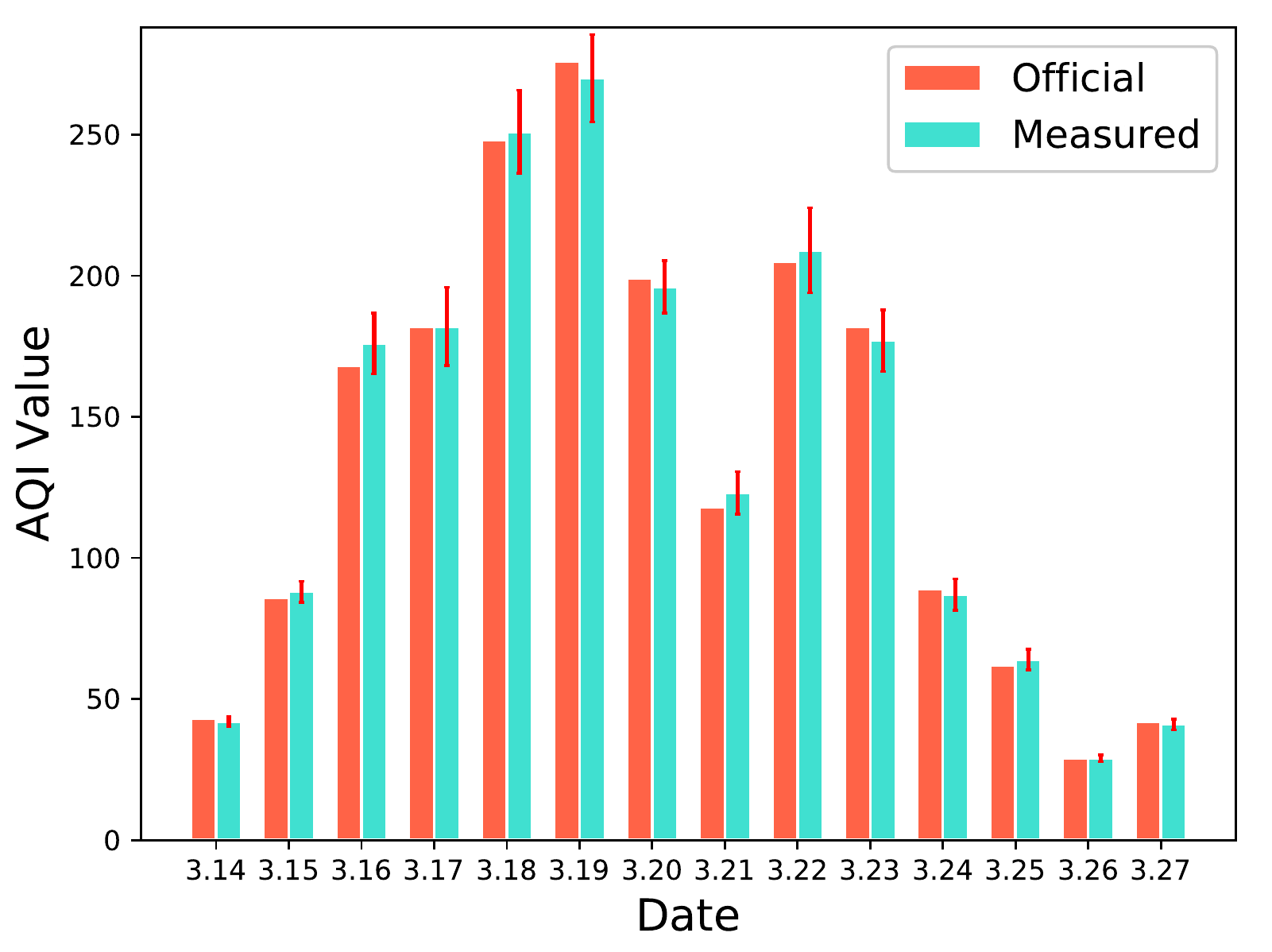}
\caption{AQI value comparison between official data and data we collected, for 14 days in March, 2017.}
\vspace{-0.15cm}
\label{aqi compare}
\end{figure}

\subsection{Data Reliability}
To verify that there is no measurement error, we show the results of the relationship between our collected data and the official data~(i.e., Haidian station~\cite{chinaweather}), in Fig.~\ref{aqi compare}. Note that the official data is limited and only for the 2D space, while our system is mobile and suitable for the 3D space profiling. We select 14 consecutive days for about 60 instances of monitoring from Mar. 14 to Mar. 27, 2017, to verify the reliability of our measurement.
We use the two-tailed hypothesis test~\cite{statistics book}: $\mathcal{H}_0: \mu_1 = \mu_2$ $vs.$ $\mathcal{H}_1: \mu_1 \neq \mu_2$, where $\mu_1$ denotes our average measured data for all days and $\mu_2$ is the average for the official ones. The test result, $P=0.9999 \gg 0.05$, indicates that there is no significant difference between the two values, which confirms the reliability of our measurements.

\subsection{Selection of Model Parameters}
According to the previous AQI monitoring results for coarse-grained scenarios~\cite{model3}, AQI is related to wind~(including speed and direction), temperature, humidity, altitude and spatial locations. But for fine-grained scenarios, correlations between AQI and these spatial parameters need to be reconsidered, due to the heterogenous diffusion in both vertical and horizontal directions in a small-scale area.
In this test, all these potential parameters are measured by our ARMS with different sensors.
To evaluate the real correlation between these parameters, we adopt the spatial regression according to~\cite{sptialmodel}, and test the coefficient for each parameter. Mathematically, the spatio-temporal model is given below:
\begin{equation}
C(s_i) = \boldsymbol{z}(s_i) \boldsymbol{\beta}^\mathrm{T} + \varepsilon(s_i),
\end{equation}
where $C(s_i)$ is the particle concentration at position $s_i$, $\boldsymbol{z}(s_i) = (z_1(s_i),...,z_n(s_i))$ denotes the vector of $n$ parameters at $s_i$, and $\boldsymbol{\beta} = (\beta_1,...,\beta_n)$ is the coefficient vector. $\varepsilon(s_i)\sim N(0,\sigma^2)$ is the Gaussian white-noise process.

Based on our data, we use the least square regression and implement a hypothesis test for each coefficient $\beta_j$, as $\mathcal{H}_0: \beta_j = 0$. The results in Table~\ref{Table 1} indicate that wind and location are highly related to AQI distribution, whereas temperature and humidity are not.
%$P_{wind} = 7.5693 \e{-5} \ll 0.05$, $P_{location} = 2.0981 \e{-5} \ll 0.05$, $P_{temperature} = 0.9070 \gg 0.05$, $P_{humidity} = 0.6996 \gg 0.05$.

\begin{table}[!t]
\caption{Result of the Hypothesis Test}
\centering
\begin{tabular}{|l|c|}
\hline
Tested Parameter & P\_value \\
\hline
Wind & $7.5693 \e{-5}\ (\ll 0.05)$ \\
\hline
Location & $2.0981 \e{-5}\ (\ll 0.05)$ \\
\hline
Temperature & $0.9070\ (\gg 0.05)$ \\
\hline
Humidity & $0.6996\ (\gg 0.05)$ \\
\hline
\end{tabular}
\label{Table 1}
\end{table}

% ==================================================Model
\section{Fine-grained AQI Distribution Model}
In this section, we provide a prediction model considering both physical particle dispersion and NN structure. We first introduce the physical dispersion model for the fine-grained scenario. Then, we provide a brief introduction of NN we adopt in modeling, which can adapt to complicated cases, such as the non-linearity introduced by extreme weather. Finally, we embed the dispersion model in NN to design our distribution model.

\subsection{Physical Particle Dispersion Model}
We first address the physical particle dispersion model for fine-grained scenarios.
Specifically, we ignore the influence of temperature and humidity according to discussions in Section 2.D, and select the Gaussian Plume Model~(GPM) in the particle movement theory~\cite{model1}, to describe the particle's dispersion. GPM is widely used to describe particles' physical motion~\cite{mathmodel,model2}, and its robustness has been proved in a small scale system~\cite{robust}.

GPM is expressed as
\begin{equation}
C(x,y,z) = \frac{Q}{2\pi \sigma_y \sigma_z u} \exp \left(-\frac{(z-H)^2}{2\sigma_z^2}\right) \exp\left(-\frac{y^2}{2\sigma_y^2}\right),
\label{GPM}
\end{equation}
where $Q$ is the point source strength, $u$ is the average wind speed, and $H$ denotes the height of source.

To adopt GPM into the fine-grained scenario, the GPM is revised as below
\begin{align}
C(\vec{x},u) & = \int_{-\frac{L}{2}}^{\frac{L}{2}} \frac{\lambda}{2\pi \sigma_y \sigma_z u} \exp\left(-\frac{(z-H)^2}{2\sigma_z^2}-\frac{y^2}{2\sigma_y^2}\right) \ud y  \nonumber\\
& = \frac{\lambda \exp\left(-\frac{(z-H)^2}{2\sigma_z^2}\right)}{2\pi \sigma_z u} \int_{-\frac{L}{2\sigma_y}}^{\frac{L}{2\sigma_y}}\exp\left(-\frac{\gamma^2}{2}\right) \ud \gamma \nonumber \\
& = \frac{\lambda}{\sqrt{2\pi} \sigma_z u} \exp\left(-\frac{(z-H)^2}{2\sigma_z^2}\right)\left[1-2Q\left(\frac{L}{2\sigma_y}\right)\right],
\label{revised model}
\end{align}
where $C(\vec{x},u)$ is the AQI value at location $\vec{x}$, $u$ is the real wind speed at different locations in the entire space, $H$ denotes a variable that reflects the influence of wind direction, which presents severely polluted areas along $z$-axis.
%$H$ is then a parameter needs to be estimated in regression model.
Pollution mainly derives as a line source aligned the $y$-axis, and $L$ denotes the length of polluted source, $\lambda$ denotes the particle density at the source. $\sigma_y$ and $\sigma_z$ are diffusion parameters in $y$ and $z$ directions, and are both empirically given.
The dispersion model in~(\ref{revised model}) can reflect physical characteristics, but can hardly deal with unpredictable complicated changes, such as the non-linearity introduced by extreme weather.

\begin{figure}[!t]
\small
\centering
\includegraphics[width=3.4in]{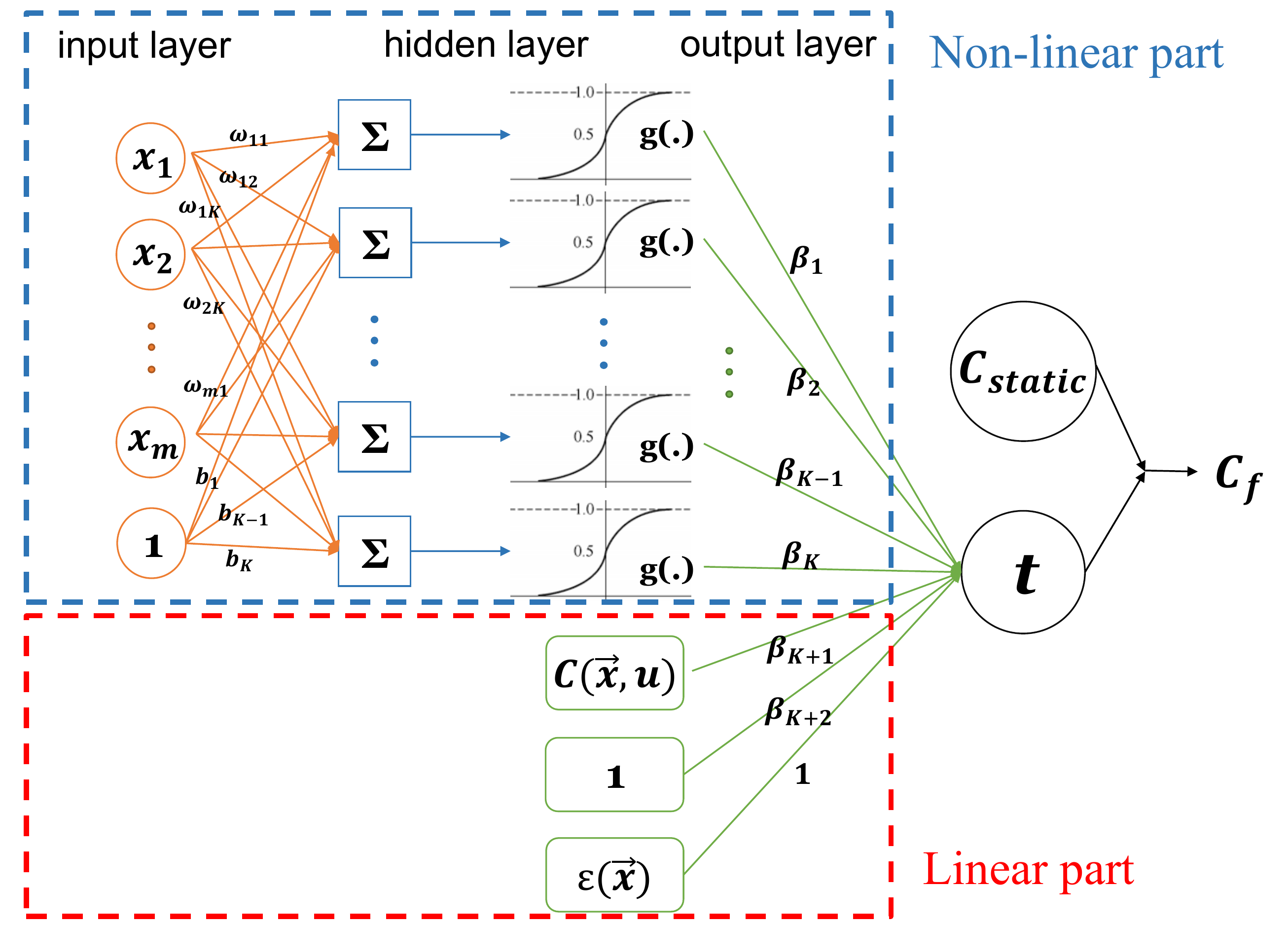}
\caption{The model structure of GPM-NN.}
\vspace{-0.15cm}
\label{figure AGNN model}
\end{figure}

\subsection{Neural Network Model}
The neural network model, especially multilayer perceptron~(MLP), has been wildly adopted to do estimation for air quality~\cite{NN1,NN2,NN3,NN4}. They usually train models by using a huge amount of data to achieve decent performance. All possible influential factors are involved as the neural network input variables for network training.
Other types of NN~\cite{Revision:CNN,Revision:RNN} are proposed for better classification with more complex structures.
As it has been proved that a three-layer neural network can compute any arbitrary function~\cite{nnproof1,nnproof2,nnproof3}, NN is able to present the complicated changes in fine-grained scenario.
However, without considering the physical characteristics of AQI, the NN model may overfit and perform worse on the test data than on the training data~\cite{NN1}.

\subsection{GPM-NN Model}
In order to utilize the advantages of both GPM and NN, we embed the revised GPM in NN, and put forward GPM embedding NN~(GPM-NN) model.
% To present non-linear factors that can influence the AQI distribution, we add NN architecture to handle such problems.

\subsubsection{Model Description}
As shown in Fig.~\ref{figure AGNN model}, the model structure contains a linear part~(the physical dispersion model) and a non-linear part~(the NN structure) for fine-grained AQI distribution, respectively. Let $N$ be the total number of data collected by ARMS, which is represented by a pair $(\bm{X}_j,t_j)$, where $\bm{X}_j=[x_1\ x_2\ \dots \ x_m]^{\rm{T}}$ is the $j^{th}$ sample with a dimensionality of $m$ variables and $t_j$ is the measured AQI value.
\begin{enumerate}[\indent (a)]
\item In the non-linear NN part, let $K$ denote the total number of neurons in the hidden layer. The weights for these neurons are denoted by $\bm{W}=[\bm{W}_1\ \bm{W}_2\ \dots \ \bm{W}_K]$, where $\bm{W}_i=[\omega_{i1}\ \omega_{i2}\ \dots \ \omega_{im}]$ is the $m$-dimensional weight vector containing the weights between the components of input vectors and the $i^{th}$ neuron in the hidden layer. $\bm{b}=[b_1\ b_2\ \dots \ b_K]$ is the bias term of the $i^{th}$ neuron. The non-linear part with $K$ neurons in the hidden layer will have $\bm{\beta}=[\beta_1\ \beta_2\ \dots \ \beta_K]$ as weights for output layer and $g(\cdot)$ is the activation function.
\item In the linear part, we use $C(\vec{x},u)$, a constant value and a Gaussian process as inputs, to reflect the influence of the physical model. The regression weights are correspondingly determined as $\beta_{K+1}$, $\beta_{K+2}$ and 1.
\end{enumerate}

Thus, the mathematical expression of the proposed model can be written as
\begin{equation}
\begin{split}
t(\vec{x},u) = & \sum_{i=1}^{K} \beta_i g(\bm{W}_i \bm{X}_j + b_i) + \beta_{K+1} C(\vec{x},u) + \\
& \beta_{K+2} + \varepsilon(\vec{x}), \qquad \quad j = 1, 2, \dots, N,
\label{final model}
\end{split}
\end{equation}
where $t(\vec{x},u)$ is the estimated value of $t_j$ and it represents the model's output. $C(\vec{x},u)$ is the output of the dispersion model in~(\ref{revised model}) and $\beta_i$ are regression coefficients. $\varepsilon(\vec{x})\sim N(0,\sigma^2)$ is the measurement error defined by a Gaussian white-noise process.
Since there is a risk that the NN part will overfit and perform worse on the test data than training data, the estimated AQI value is expressed as
\begin{equation}
C_f(\vec{x},u) = C_{static} + t(\vec{x},u),
\end{equation}
where $C_{static}$ is the average value of our measured AQI in a day, which is an invariant to quantify basic distribution characteristics.

\subsubsection{Parameter Estimation}
As shown in~(\ref{final model}), GPM-NN has $(K$+$3)$ parameters, $H$, $\beta_1$, $\beta_2$, $\dots$, $\beta_{K+2}$, which need to be estimated based on data collected by ARMS. 50 days' data are used for training the non-linear part of GPM-NN.
We use the least square regression to estimate the parameters.
Let $\mathcal{S}$ denote the residual error as
%\vspace{0.3cm}
\begin{equation}
\mathcal{S} = \sum_{i=1}^{N} \left\| \hat{C}_f(\vec{x}_i,u_i)- \beta_{K+2} - \beta_{K+1} C(\vec{x},u) - \sum_{j=1}^{K} \beta_j g_j \right\|^2
\label{residual}
\end{equation}
where $i$ denotes the measuring sample of the $i^{th}$ observation point, and $g_j=g(\bm{W}_j \bm{X}_i + b_j)$.
% $\beta_0^{'}=C_{static}+\beta_0$ and $\beta_1^{'}=(1-2Q(\frac{L}{2\sigma_y}))\beta_1\lambda/\sqrt{2\pi}$.

% ==================================================Proposition and Proof
\begin{proposition}
Equation (\ref{residual}) has a unique minimum point for estimated parameters $\beta_{1}$, $\beta_{2}$, $\dots$, $\beta_{K+2}$ and $H$, when $\sigma_z^2>\max\{2z_i^2,2H_0^2\}$.
\end{proposition}

\begin{proof}
See Appendix A.
\end{proof}
\vspace{0.3cm}

To find the minimum point of the residual error function $\mathcal{S}(H,\beta_1,\dots,\beta_{K+2})$, we use the Newton method~\cite{newton method} to solve the following equations whose analytical solution does not exist, as
\begin{equation}
\left\{
\begin{aligned}
\begin{split}
& \frac{\partial\mathcal{S}}{\partial H} = 0, \\[7.5pt]
& \frac{\partial\mathcal{S}}{\partial \beta_j} = 0,\qquad j = 1, 2, \dots, K+2.
\end{split}
\end{aligned}
\right.
\label{multivar}
\end{equation}

% ==================================================Matrix Solution
When the estimation value of $H$~(denoted as $H^{*}$) is determined, $C(\vec{x},u)$ is correspondingly determined. Denote
\begin{equation*}
\bm{J} =
\left[
\begin{matrix}
\begin{smallmatrix}
 g(\bm{W}_1 \bm{X}_1 + b_1)  & \cdots & g(\bm{W}_K \bm{X}_1 + b_K)  & C(\vec{x}_1,u_1)  & 1  \\[7pt]
 g(\bm{W}_1 \bm{X}_2 + b_1)  & \cdots & g(\bm{W}_K \bm{X}_2 + b_K)  & C(\vec{x}_2,u_2)  & 1  \\[7pt]
 \vdots & \ddots & \vdots & \vdots & \vdots  \\[7pt]
 g(\bm{W}_1 \bm{X}_N + b_1)  & \cdots & g(\bm{W}_K \bm{X}_N + b_K)  & C(\vec{x}_N,u_N)  & 1  \\
\end{smallmatrix}
\end{matrix}
\right]_{ N \times (K+2) }
\end{equation*}
as the model output matrix, and similarly
\begin{equation*}
\bm{\beta} =
\left[
\begin{matrix}
\begin{smallmatrix}
 \beta_1  \\[5pt]
 \beta_2  \\[5pt]
 \vdots   \\[5pt]
 \beta_K  \\[5pt]
 \beta_{K+1}  \\[5pt]
 \beta_{K+2}  \\[5pt]
\end{smallmatrix}
\end{matrix}
\right]_{ (K+2) \times 1 }
\end{equation*}
is the vector that needs to be estimated. Hence, the estimated value of $N$ samples can be written as
\begin{equation}
\bm{T} = \bm{J} \bm{\beta}.
\end{equation}

Note that $\bm{J}$ is both row-column full rank matrix, which has a corresponding generalized inverse matrix~\cite{inverse matrix}.
As we have proved (\ref{residual}) has a unique minimum point, we then have
\begin{equation}
\begin{split}
\bm{\beta} & = (\bm{J}^{\rm{T}} \bm{J})^{-1} \bm{J}^{\rm{T}} \bm{J} \bm{\beta} \\
& = (\bm{J}^{\rm{T}} \bm{J})^{-1} \bm{J}^{\rm{T}} \bm{T} \\
& = \bm{J}^{\dag} \bm{T} ,
\end{split}
\end{equation}
where $\bm{J}^{\dag} = (\bm{J}^{\rm{T}} \bm{J})^{-1} \bm{J}^{\rm{T}}$ is known as the Moore-Penrose pseudo inverse of $\bm{J}$. This equation is the least squares solution for an over-determined linear system and is proved to have the unique minimum solution~\cite{moore penrose}. Thus, this equation is equal to the multivariate equation in (\ref{multivar}), by which we can find the minimum value point of $\mathcal{S}$.

\subsubsection{Performance Evaluation}
To determine the initial value of the weights $\bm{W}$ and biases $\bm{b}$ for the hidden layer, we use the training data to do preprocessing and acquire the optimal values. Hence, the model can be completely determined for describing the AQI distribution in fine-grained scenarios.

For evaluating the performance of GPM-NN, we use average estimation accuracy~(AEA) as the merit, expressed as
\vspace{0.3cm}
\begin{equation}
\overline{AEA} = \frac{1}{n} \sum_{i=1}^{n} \left( 1 - \frac{| \hat{C}_f(i) - C_f(i) |}{C_f(i)} \right),
\label{aea}
\vspace{0.3cm}
\end{equation}
where $n$ denotes the total locations in the scenario, $\hat{C}_f(i)$ denotes the estimation AQI value in the $i^{th}$ location and $C_f(i)$ denotes the real measured value.
In Section V and Section VI, we compare the accuracy of AQI map constructed by our GPM-NN and other existing models.

\vspace{0.5cm}
% ==================================================UAV System
\section{Adaptive AQI Monitoring Algorithm}
In this section, we provide the adaptive monitoring algorithm of ARMS.
Intuitively, a larger number of measurement locations introduce a higher accuracy of the AQI map. However, based on the physical characteristic of particle dispersion in GPM-NN, we can build a sufficiently accurate AQI map by regularly measuring only a few locations. This process can effectively save the energy, and thus improve the efficiency of the system.
Specifically, an AQI monitoring is decomposed into two steps---$complete\ monitoring$ and $selective\ monitoring$---for efficiency and accuracy.
We first trigger $complete\ monitoring$ everyday for one time, to establish a baseline distribution.
Then ARMS periodically~(e.g., every one hour) measures only a small set of observation points, which are acquired by analysing the characteristic of the established AQI map. This process, named as $selective\ monitoring$, is based on GPM-NN to update the realtime AQI map. By accumulating current measurements with the previous map, a new AQI map is generated timely.
Every time when selective monitoring is done, ARMS compares the newly-measured results and the most recent measurement. If there is a large discrepancy between them, which indicates that the AQI experiences severe environmental changes, we would again trigger the complete monitoring to rebuild the baseline distribution.
Thus, ARMS can effectively reduce the measurement effort as well as cope with the unpredictable spatio-temporal variations in the AQI values.

\begin{figure*}[!t]
\small
\centering
\includegraphics[width=0.7\textwidth]{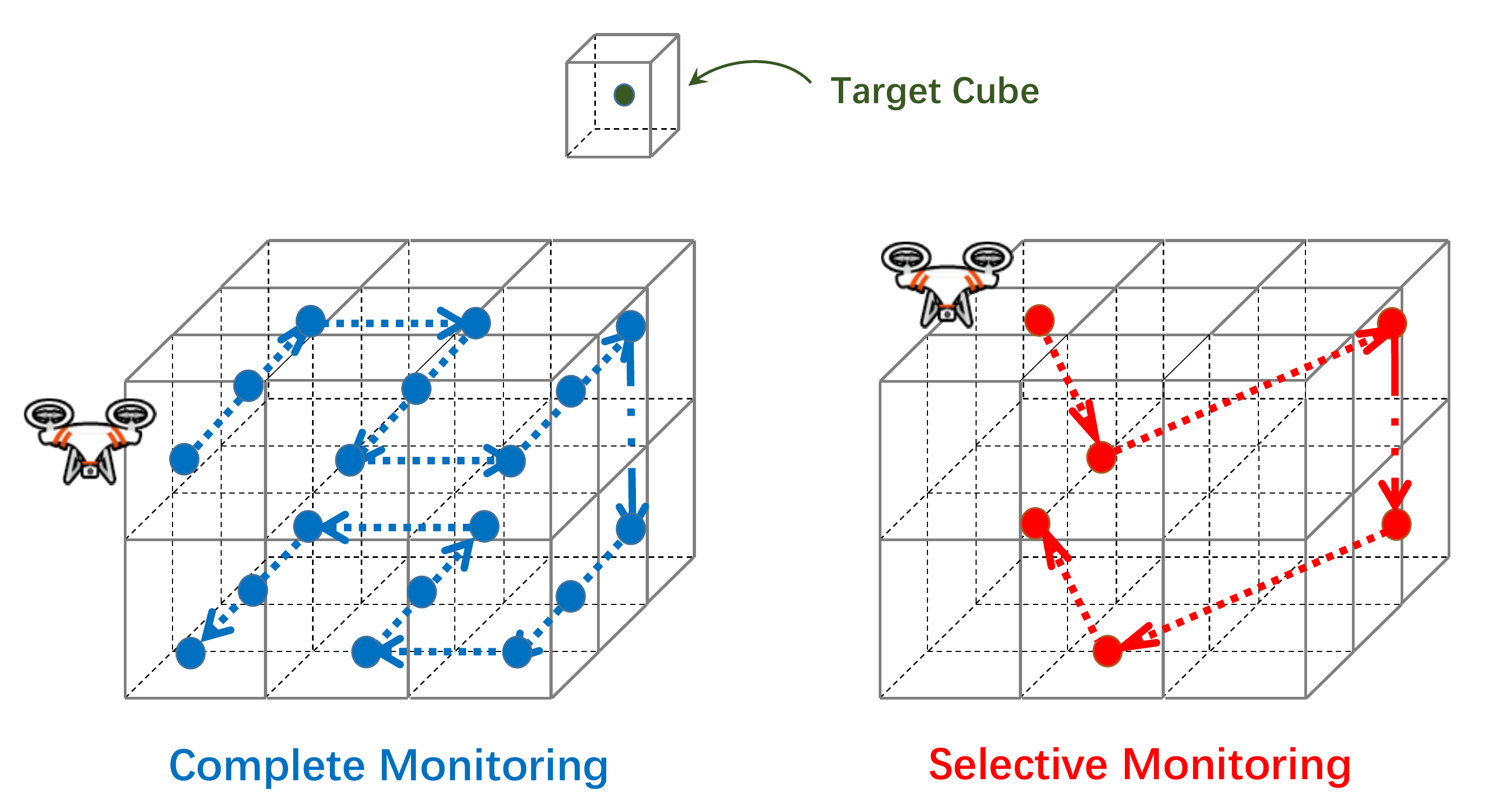}
\caption{An example of the adaptive monitoring algorithm, i.e, complete and selective monitoring.}
\vspace{-0.15cm}
\label{monitoring}
\end{figure*}

\subsection{Complete Monitoring}
The $complete$ $monitoring$ is designed to obtain a baseline characteristic of the AQI distribution in a fine-grained area and is triggered at a day interval.
%The complete monitoring performs comprehensive AQI surveys to collect ``baseline'' or repeatable AQI-condition information. Such repeatable information is important, as it provides a baseline profile of spatial AQI distribution and is useful for the design of efficient monitoring techniques---selective monitoring.

The entire space can be divided into a set of $5$m$\times$$5$m$\times$$5$m cubes.
In the complete monitoring process, ARMS measures all cubes continuously and builds a baseline AQI map using GPM-NN. The process is of high dissipation, and thus is triggered over a long observation period.

% \vspace{0.25cm}
% =======================================System Algorithm
\begin{algorithm}[!t]
\caption{Operation of monitoring algorithm}
/* \textbf{\emph{Complete Monitoring:}} triggered between days */\\
%\vspace{-1.5mm}
\For{$i=1$ to $sum(Cube)$}
{%\vspace{-1.5mm}
    measure the AQI value of $Cube_{i}$ and record\;%\vspace{-1.5mm}
    move to the next cube\;%\vspace{-1.5mm}
}
%\vspace{-1.5mm}
generate baseline 3D AQI map $\mathbb{B}$\;
\vspace{0.3cm}
/* \textbf{\emph{Selective Monitoring:}} triggered between hours */\\
%\vspace{-1.5mm}
\For{$i=1$ to $sum(Cube)$}
{%\vspace{-1.5mm}
    calculate $\textit{PDT}_{cubei}$\;%\vspace{-1.5mm}
    \If{$\textit{PDT}_{cubei} \geq \textit{PDT} \ \ |\ |\ \ \textit{PDT}_{cubei} \leq \delta$}
    {%\vspace{-1.5mm}
        add $Cube_{i}$ to $\mathcal M$\;%\vspace{-1.5mm}
    }%\vspace{-1.5mm}
}
%\vspace{-1.5mm}
generate min trajectory $\mathbb{D}$ of $\mathcal M$\;%\vspace{-1.5mm}
\ForAll{$p_{i} \in \mathbb{D}$}
{%\vspace{-1.5mm}
    measure the AQI value of $Cube_{i}$ and record\;%\vspace{-1.5mm}
}
%\vspace{-1.5mm}
update the realtime AQI map $\mathbb{M}$ based on previous $\mathbb{B}$ and $\mathbb{D}$\;%\vspace{-1.5mm}
\If{$\mathbb{M}$ deviates $\mathbb{B}$ by a large $\sigma$}
{%\vspace{-1.5mm}
    enter the complete monitoring period\;%\vspace{-1.5mm}
}
%\vspace{-1.5mm}
\end{algorithm}

\subsection{Selective Monitoring}
To reflect changes of the AQI distribution in a small-scale space over time~(e.g., between each hour in a day)~\cite{blueaer}, ARMS uses the $selective$ $monitoring$ to capture such dynamics.
The selective monitoring makes use of previous AQI map, by analyzing the physical characteristics of it, to reduce the monitoring overhead in the next survey and maintain the realtime AQI map accordingly.

In the selective monitoring process, ARMS measures AQI value of only a small set of selected cubes and generates AQI map over the entire fine-grained area.
To deal with the inherent tradeoff between measurement consumption and accuracy, we put forward an important index called the partial derivative threshold~(\textit{PDT}), to guide system selecting specific cubes. $\textit{PDT}$ is defined as
\vspace{0.3cm}
\begin{equation}
PDT_i = \frac{\left| \frac{\partial C_f}{\partial x_i} \right| - \left| \frac{\partial C_f}{\partial x_i}\right|_{min}} {\left|\frac{\partial C_f}{\partial x_i}\right|_{max}-\left|\frac{\partial C_f}{\partial x_i}\right|_{min}},
\vspace{0.2cm}
\end{equation}
where $x_i$ denotes the $i^{th}$ variable in GPM-NN~($i=1,2,\dots,m$), and $C_f = C_f(\vec{x},u)$ denotes the entire distribution in a small-scale area. $|\partial C_f/\partial x_i|_{min}$ and $|\partial C_f/\partial x_i|_{max}$ denote the minimum and the maximum value of the partial derivative for parameter $x_i$, respectively. Note that $\partial C_f/\partial x_i$ describes the upper bound of dynamic change degrees we can tolerate, expressed as
\begin{align}
\frac{\partial C_f}{\partial x_i} = PDT_i\ \cdot & \left(\left|\frac{\partial C_f}{\partial x_i} \right|_{max} - \left|\frac{\partial C_f}{\partial x_i}\right|_{min} \right) + \left|\frac{\partial C_f}{\partial x_i}\right|_{min}, \nonumber \\[5pt]
& \qquad 0 \leq PDT_i \leq 1.
\label{PDT1}
\end{align}

For each parameter, there is one corresponding \textit{PDT}. In general, \textit{PDT} reflects the threshold for dynamic change degrees in a fine-grained area. Area that has large change rate of model's parameters would have a larger \textit{PDT} value, indicating more drastic changes. When given a specific \textit{PDT}, any cube whose $\partial C_f/\partial x_i$ is above threshold of (\ref{PDT1}) will be moved into a set $\mathcal M$. Moreover, when $\textit{PDT}_i$ is too small~(less than a small const $\delta$), the corresponding $i^{th}$ cube will also be added into $\mathcal M$.
Mathematically, set $\mathcal M$ is given as
\vspace{0.15cm}
\begin{equation}
\mathcal M = \{i \mid PDT_i \geq PDT\} \cup \{i \mid PDT_i \leq \delta \}.
\vspace{0.15cm}
\end{equation}

\begin{remark}
Elements in $\mathcal M$ can be the severe changing areas in a small-scale space~(e.g., a tuyere or abnormal building architecture), or typically the lowest or the highest value that can reflect basic features of the distribution. These elements are sufficient to depict the entire AQI map, and hence are needed to be measured between two measurements. Thus, by only measuring cubes in $\mathcal M$, ARMS can generate a realtime AQI map implemented by GPM-NN, while greatly reducing the measurement overhead.
\end{remark}

In general, \textit{PDT} is adjusted manually for different scenarios.
When \textit{PDT} is low, the threshold for abnormal cubes declines, indicating the measuring cubes will increase and the estimation accuracy is relatively high. However, it can cause great battery consumption. On the other hand, as \textit{PDT} is high, the measuring cubes will decrease. This can cause a decline in accuracy, but can highly reduce consumption. In summary, the tradeoff between accuracy and consumption should be studied to acquire a better performance of whole system.

% ==================================================Optimization
\subsection{Trajectory Optimization}
When target cubes in set $\mathcal M$ are determined, the total network can be modelled as a 3D graph $G=(V,E)$ with a number of $|V|$ target cubes. Hence, finding the minimum trajectory over these cubes is equal to find the shortest hamiltonian cycle in a 3D graph. This problem is known as the traveling salesman problem~(TSP), which is NP-hard~\cite{TSP}.

To solve TSP in this case, we propose a greedy algorithm to find the sub-optimal trajectory.
In the fine-grained scenario, ARMS has power consumption and can monitor no more than $n$ cubes over one measurement. To find the corresponding trajectory, we focus on how to determine the next measuring cube based on current location of ARMS.
Let $\mathbb{Z} = \{O_0,O_1,...,O_{|V|-1}\}$ be the set of coverage cubes, with $O_i$ denotes every observation cube.
The aim is to acquire as many target cubes as possible over the trajectory for higher AQI estimation accuracy.
Considering the significant physical characteristic of $\textit{PDT}$ above, our greedy solution can be formulated as: maximize the next cube's $\textit{PDT}$, as well as minimize the traveling cost from current location to next cube.
Hence, finding the optimal trajectory in this case is equal to an iteration of solving the following optimization problem, expressed as
\vspace{0.15cm}
\begin{equation}
\begin{split}
i^* = & \arg \max \limits_{i} \  \Bigg| \frac{PDT_i}{cost(i)} \Bigg| \\
s.t. \ \ & O_i \in \mathcal M , \\
& O_i  \cap  \bigcup \{ O_{0},O_{1},\dots,O_{i-1} \} = \varnothing , \\
\end{split}
\label{optimization problem}
\end{equation}
where $cost(i)$ is the consumption for the UAV to traverse from the $(i-1)^{th}$ cube to the $i^{th}$ cube, and $PDT_i$ is acquired by analysing the characteristic of latest AQI map.

For every current location $i$, the selection of next target cube follows (\ref{optimization problem}). Note that there are limited target cubes in $\mathcal M$, which are also determined by (\ref{PDT1}), hence the objective function aims to generate trajectory point-by-point.
Thus, using the solution of (\ref{optimization problem}), the greedy algorithm can effectively select key cubes and generate the suboptimal trajectory for ARMS in different scenarios, respectively.

For analyzing the complexity of our algorithm, there are $V$ target cubes in total that need to be added from $\mathcal M$. When current location of ARMS is at the $i^{th}$ cube, it needs to compare another $|V-i|$ edges in $G$ to determine the next measuring cube. Note that every target cube contains $m$ parameters~($m=4$ in our model), and $O(V)=O(n)$.
Thus, the total operation time is $O\left(m\sum_{i=1}^{V-1} |V-i|\right) = O(n^2)$.

Algorithm 1 describes the whole process of the monitoring algorithm. Complete monitoring is triggered between days and selective monitoring is triggered between hours.
When the monitoring area experiences severe environmental changes such as the gale, ARMS compares the result of map built by selective monitoring and the map built last time. If there is a large deviation $\sigma$ between them, ARMS would again trigger the complete monitoring to rebuild the baseline distribution.
% Note that this situation is rare, on the condition of severely polluted days.

\vspace{0.4cm}
% ==================================================Scenario I
\section{Application Scenario I: Performance Analysis in Horizontal Open Space}
In this section, we implement the adaptive monitoring algorithm in a typical 2D scenario, namely the horizontal open space. We present performance analysis of GPM-NN and adaptive monitoring algorithm in this typical scenario, respectively.

\subsection{Scenario Description}
When the 3D space has a limited range in height, ARMS needs to cover target cubes nearly in the same horizontal plane. Two distant cubes at the same height may have a low correlation, as the wind may create different concentration of pollutants in a horizontal plane.
This scenario is commonly considered as a typical 2D scenario and often with a horizontal-open space~(e.g., a roadside park), as shown in Fig.~\ref{system model1}.

%In this case, we adopt the monitoring techniques to select target cubes that follows the greedy selection rule (\ref{optimization problem}) in the same plane horizontally. During one measurement, Arms prefers selecting more target cubes from the current level, instead of moving to upper levels.

\begin{figure}[!t]
\small
\centering
\includegraphics[width=3.4in]{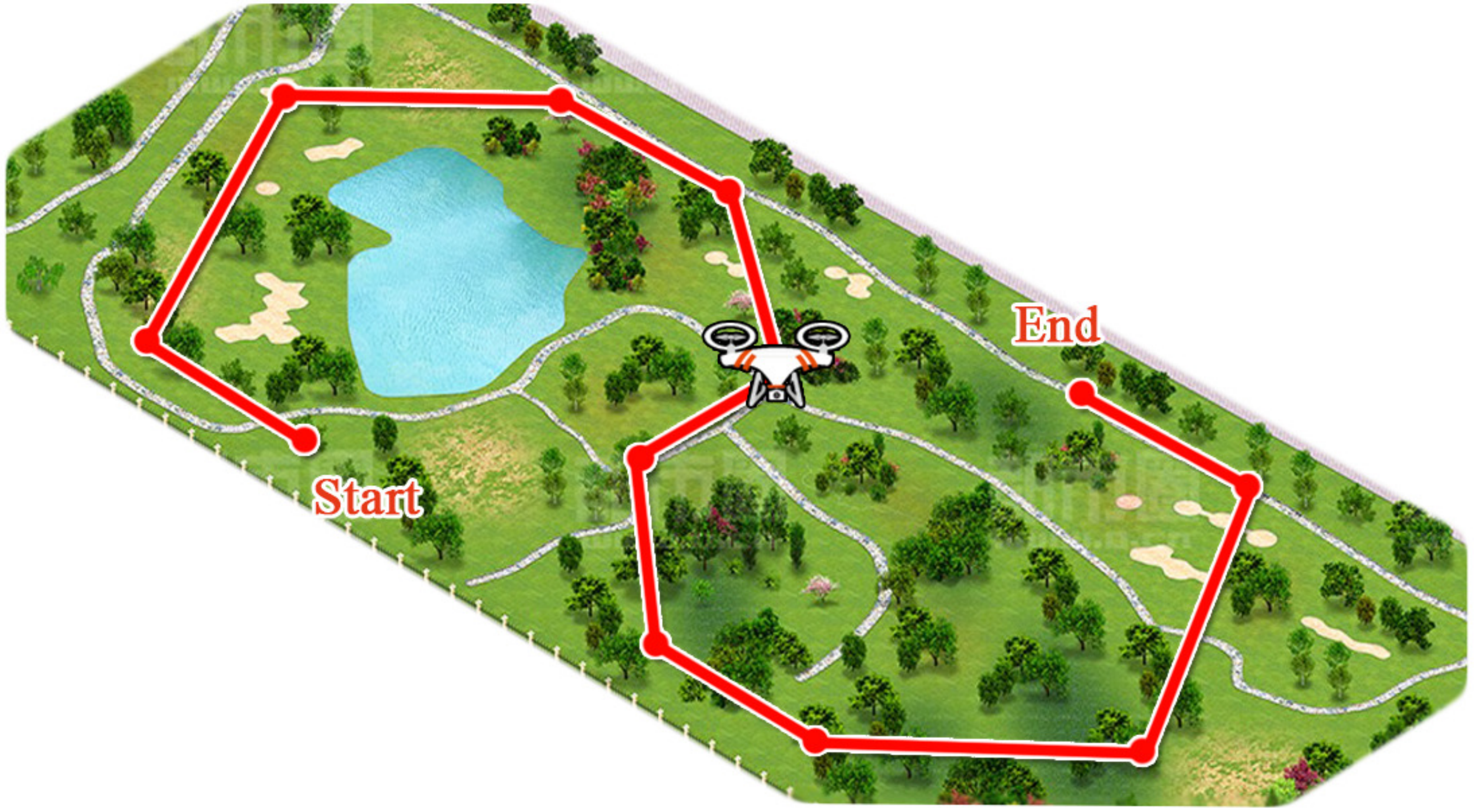}
\caption{The typical application scenarios of ARMS in 2D space~(a roadside park).}
%\vspace{-0.15cm}
\label{system model1}
\end{figure}

\iffalse
\vspace{0.2cm}
% =======================================Horizontal Greedy Search Algorithm
\begin{algorithm}[h]
\caption{Horizontal greedy search algorithm}
\begin{algorithmic}[1]
\REQUIRE $G$, $V$, $n$, ${\mathcal M}$.
\ENSURE trajectory $T$.
\STATE Randomly select a cube from ${\mathcal M}$ at ground level;
\STATE $n$ = max($V$, $n$);
\FOR{$i=1$ to $n$}
\IF{$\exists j$ at current level $s.t.\ O_j\in \mathcal{M}$}
\STATE Select an cube that maximize $|\frac{PDT_i}{cost(i)}|$ from $\{O_j\}$;
\ELSE
\STATE UAV moves up by one level;
\ENDIF
\STATE Add selected cube into $T$;
\ENDFOR
\RETURN $T$;
\end{algorithmic}
\end{algorithm}
\fi

\subsection{Performance Analysis}
In this section, we first compare the accuracy of GPM-NN with other existing models by the experimental result in Fig.~\ref{accuracy 2d}.
Then, Fig.~\ref{layernum 2d} illustrates the influence by different numbers of neurons in the hidden layer.
To study GPM-NN's performance when AQI varies, in Fig.~\ref{different aqi 2d}, we show the relationship between different AQI values and corresponding estimation accuracy.
In Fig.~\ref{HGS}, we present the performance of our monitoring algorithm versus other selection algorithms.
Finally, Fig.~\ref{tradeoff 2d} shows the tradeoff between system battery consumptions and estimation accuracy via different $\textit{PDT}$s.

\subsubsection{Model Accuracy}
In Fig.~\ref{accuracy 2d}, we compare three prediction models, our regression model GPM-NN, linear interpolation~(LI)~\cite{LI} and classical multi-variable linear regression~(MLR)~\cite{sptialmodel}, respectively, versus different values of $\textit{PDT}$.
LI uses interpolation to estimate the AQI value of undetected cubes by other measured cubes, while MLR uses multiple parameters~(e.g., wind, humidity, temperature, etc.) of measured cubes to do regression and estimation.

In the horizontal open space scenario, we can find that GPM-NN achieves the highest accuracy.
In each curve, we can see that the average estimation accuracy decreases as the $\textit{PDT}$ value increases. As discussed in Section IV-B, when $\textit{PDT}$ has a higher threshold, target cubes in set $\mathcal M$ decline, i.e., the total cubes measured by ARMS become fewer. Thus, the estimation accuracy correspondingly drops.
When $\textit{PDT}=0.1$, GPM-NN performs the best among three models, which proves the robust and precision of our model. Moreover, as $\textit{PDT}$ increases~(e.g., $\textit{PDT} = 0.75$), GPM-NN still maintains a high accuracy~(almost $80\%$), while others experience a rapid decrease. This implies that our model is suitable for adaptive energy saving monitoring in a fine-grained area.

\begin{figure}[!t]
\small
\centering
\includegraphics[width=3.3in]{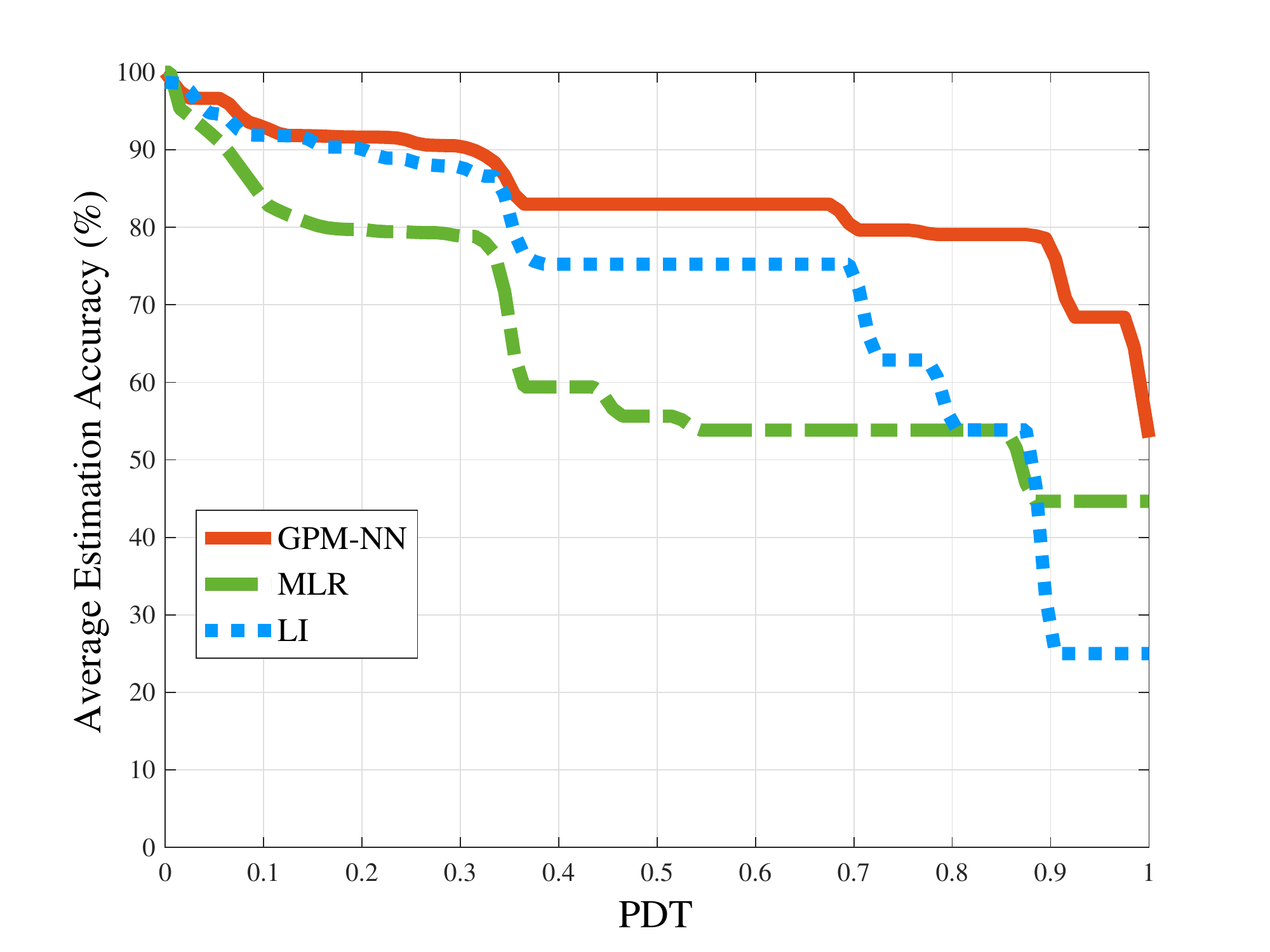}
\caption{The comparison of estimation accuracy between GPM-NN, MLR and LI, in 2D scenario.}
%\vspace{-0.15cm}
\label{accuracy 2d}
\end{figure}

\subsubsection{Effects of Neuron Numbers}
As we adopt the NN structure to introduce the non-linear part for our GPM-NN model, the number of neurons in the hidden layer can have great impacts on estimation results. In Fig.~\ref{layernum 2d}, we plot the estimation accuracy of different number of neurons in GPM-NN via $\textit{PDT}$, to study their influence.

From Fig.~\ref{layernum 2d}, when $\textit{PDT}<0.1$, the monitoring contains all cubes. When the number of neurons is 0, our model is equal to the physical model in~(\ref{revised model}) with regression, which only contains the linear part. By comparing this curve with others, we can find out that the number of neurons $=$ 0 is worse than the number of neurons $\neq$ 0. By adding the non-linear part~(NN structure), GPM-NN performs better with higher accuracy.
Moreover, the curve with fewer number of neurons~(e.g., the number of neurons $=$ 10) performs worse than with more neurons~(e.g., the number of neurons $=$ 500). In this scenario, we can find that the number of neurons $=$ 1000 can achieve the highest estimation accuracy.
We ignore the situation where the number of neurons $>$ 1000, as too many neurons in the hidden layer can cause overfitting.

\begin{figure}[!t]
\small
\centering
\includegraphics[width=3.3in]{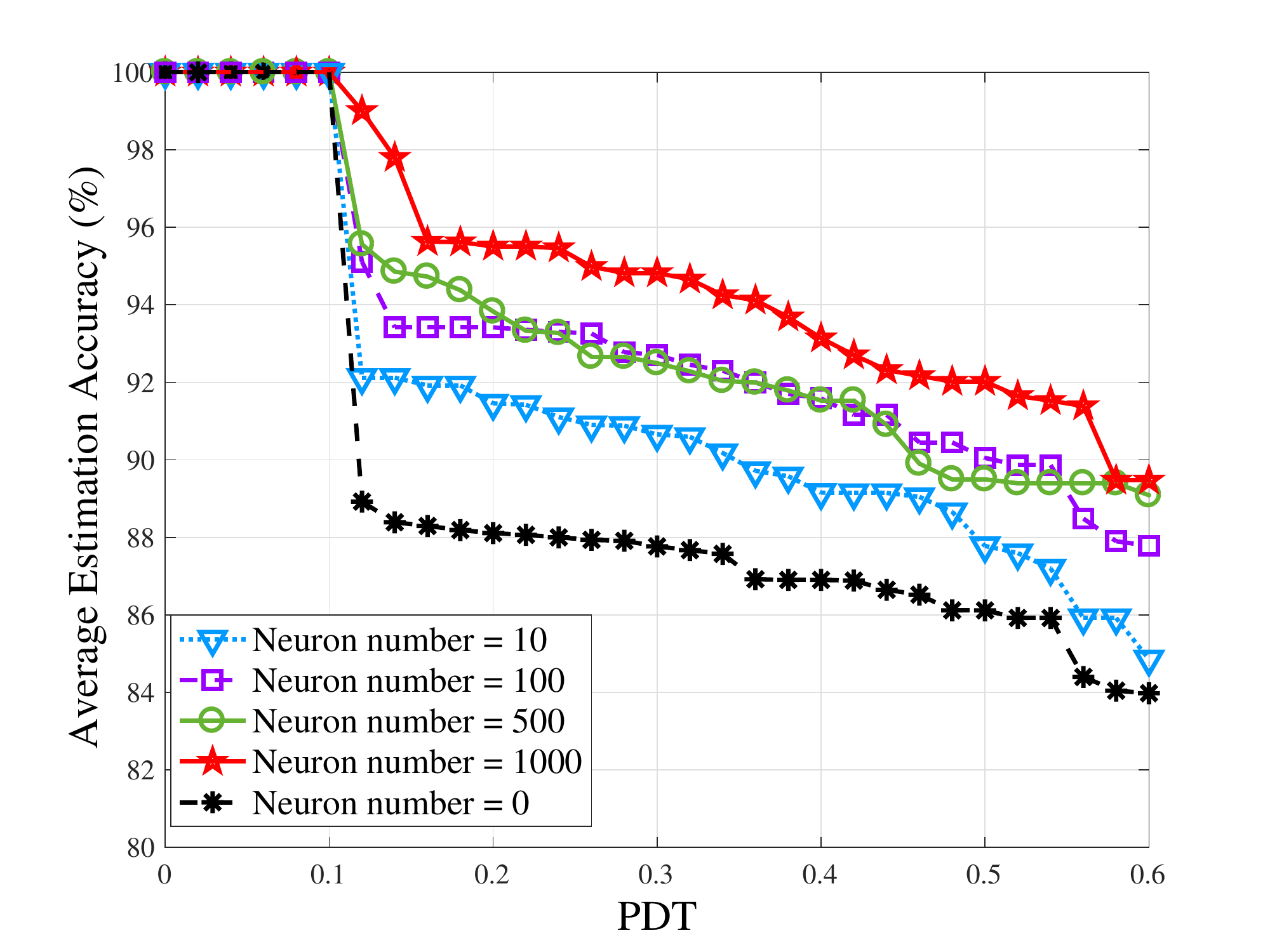}
\caption{The impact of the number of neurons in the non-linear part, in 2D scenario.}
%\vspace{-0.15cm}
\label{layernum 2d}
\end{figure}

\subsubsection{Effects of Various AQI}
In Fig.~\ref{different aqi 2d}, we plot the estimation accuracy of GPM-NN with different AQI values~(i.e., AQI $\le$ 50, 50 $\le$ AQI $\le$ 200 and AQI $\ge$ 200~\cite{chinaweather}), via different $\textit{PDT}$s.
From the curves, we can find that in 2D scenario, GPM-NN performs the best when AQI $\ge 200$. As $50\le$ AQI $\le 200$, GPM-NN also maintains high accuracy, while relatively worse when AQI is low. This indicates that our model is better predicting in moderately and highly polluted days, which has great instructing significance in forecasting severe pollution as well as prevention. This characteristic is also suitable for the adaptive monitoring algorithm when AQI is high.
Note that even GPM-NN performs not so good when AQI is low, it still outperforms other models.

\begin{figure}[!t]
\small
\centering
\includegraphics[width=3.3in]{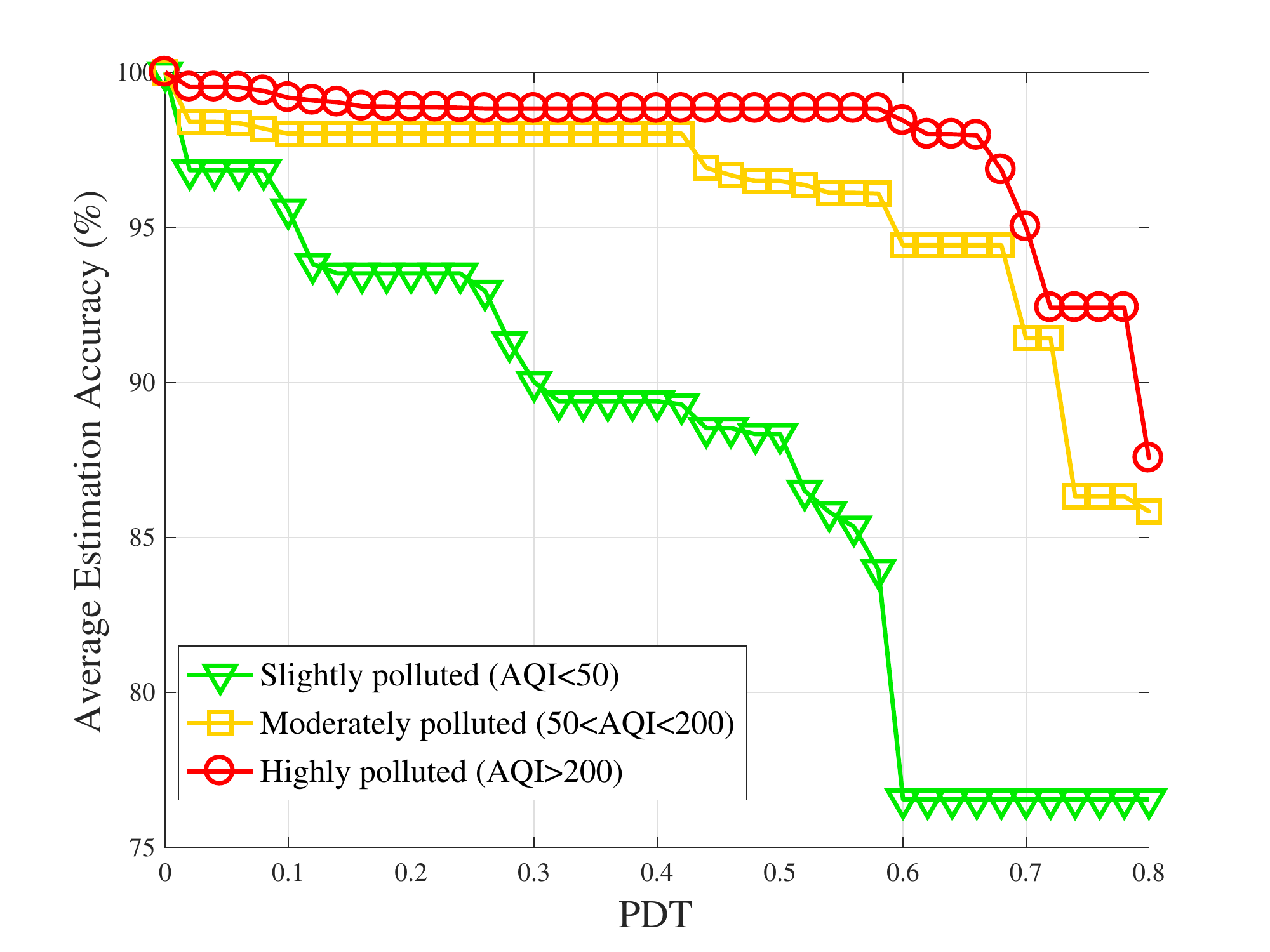}
\caption{The performance of GPM-NN with different AQI value, in 2D scenario.}
%\vspace{-0.15cm}
\label{different aqi 2d}
\end{figure}

\begin{figure}[!t]
\small
\centering
\includegraphics[width=3.3in]{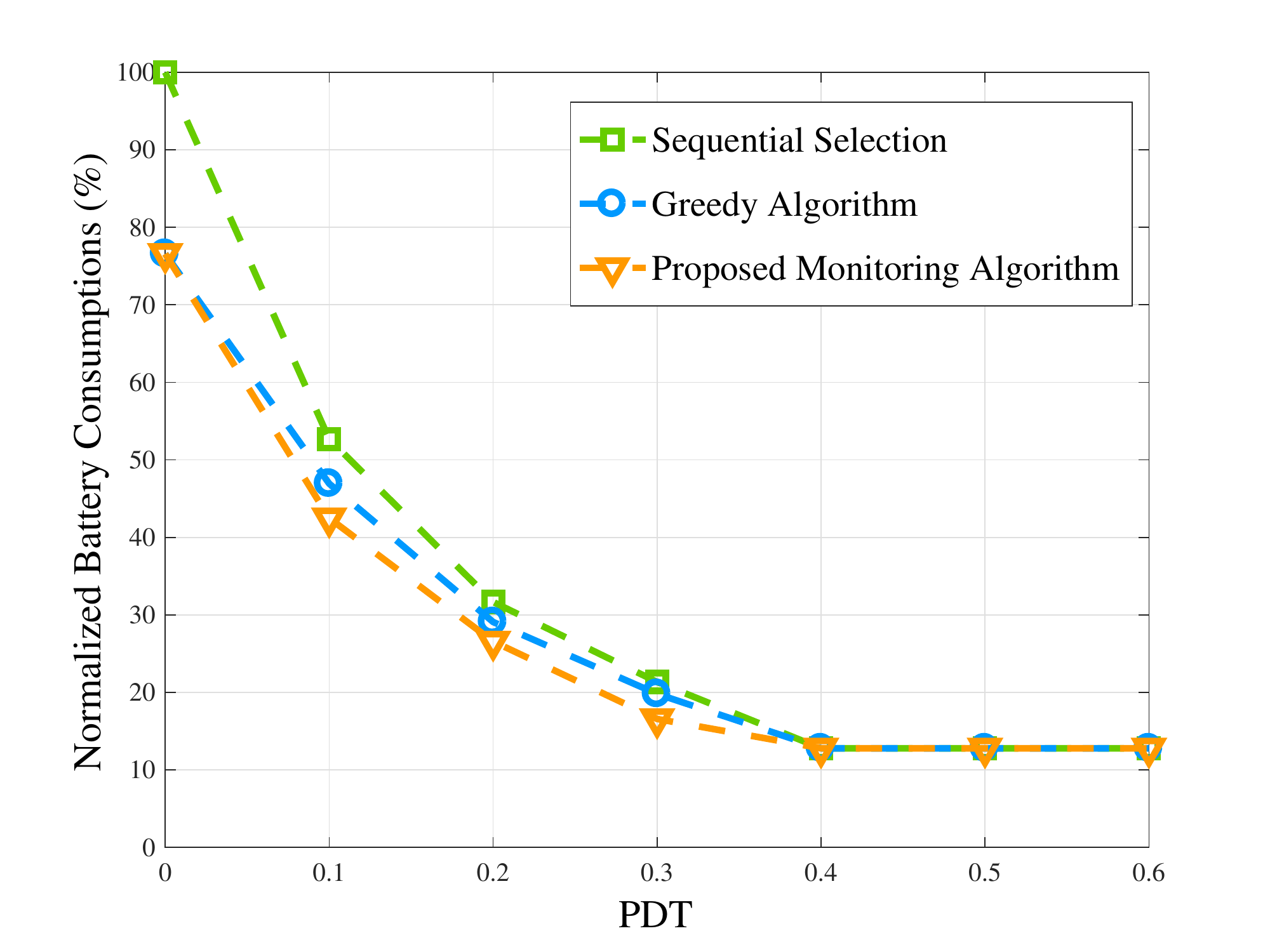}
\caption{Comparison of the Adaptive Monitoring Algorithm, Greedy Algorithm and Sequential Selection.}
%\vspace{-0.15cm}
\label{HGS}
\end{figure}

\subsubsection{Performance of Adaptive Monitoring Algorithm}
In this part, we compare the results of the proposed monitoring algorithm for trajectory planning, versus other algorithms such as greedy algorithm and sequential selection, by plotting their battery consumptions over one measurement in Fig.~\ref{HGS}. The greedy algorithm aims to select the nearest target cube in $\mathcal M$ to generate the trajectory~\cite{node3}, while sequential selection is done by selecting cubes from the bottom~(or left) to the top~(or right) in order~\cite{blueaer}.

In the typical horizontal open space, we plot the normalized battery consumption achieved by three algorithms in Fig.~\ref{HGS}, via different $\textit{PDT}$s. The normalized consumption is the cost percentage achieved by each monitoring method of one total battery charge~(i.e., 15 minutes). As $\textit{PDT}$ increases, the consumption would correspondingly decrease, as the target cubes in $\mathcal M$ would be fewer. By comparing three curves, we can see that sequential selection is the most consuming method. Our monitoring algorithm performs the best and is better than the normal greedy algorithm, while $0.1\leq PDT\leq 0.4$. After $\textit{PDT}$ reaches 0.4, the consumption of three methods becomes equal, since the target cubes in $\mathcal M$ now is so few that there is no difference in using these algorithm. Hence, the adaptive monitoring algorithm can relatively reduce the power consumption for monitoring AQI in the 2D scenario.

\begin{figure}[!t]
\small
\centering
\includegraphics[width=3.3in]{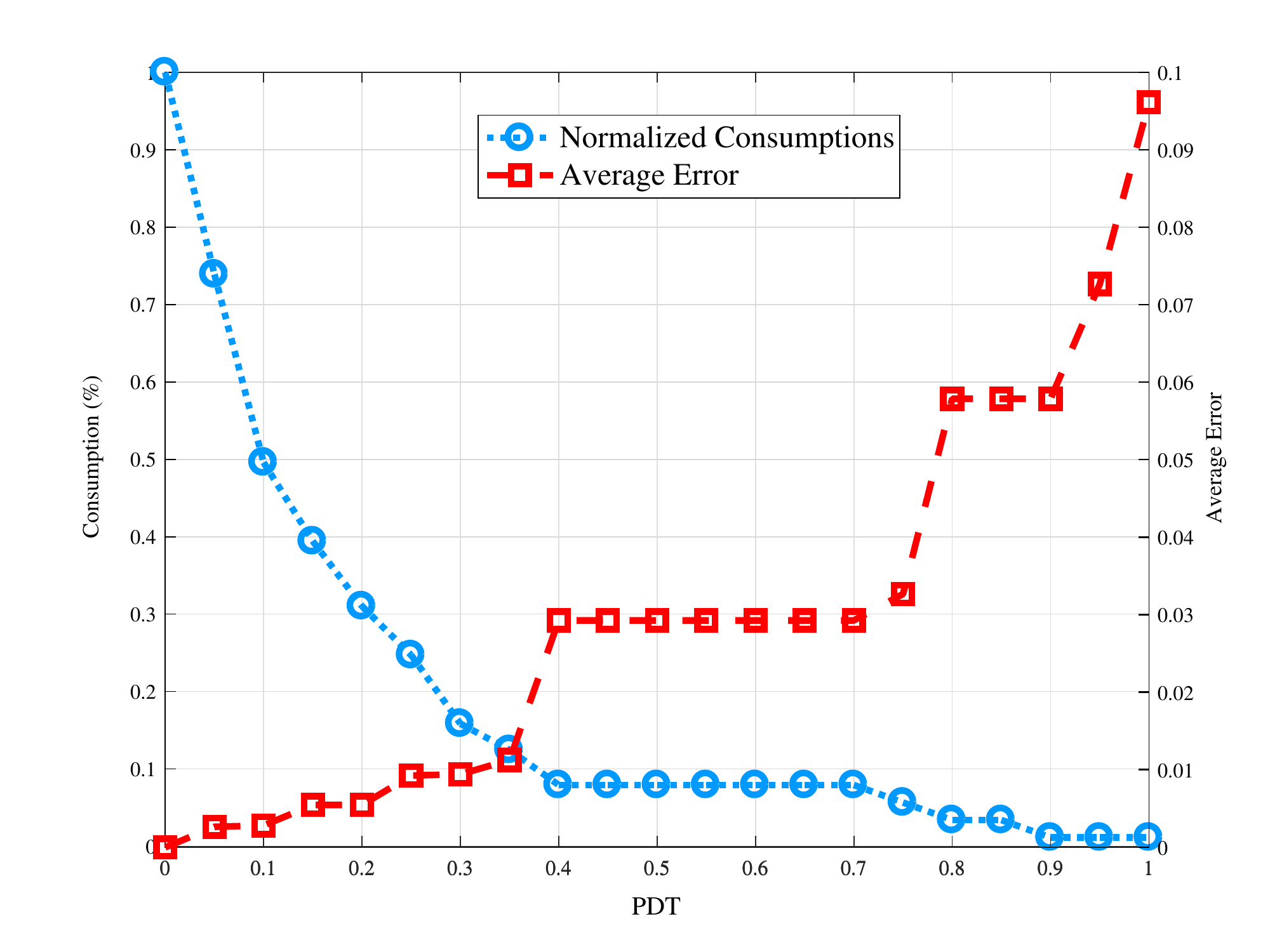}
\caption{The tradeoff between system battery consumption and estimation accuracy, in 2D scenario.}
%\vspace{-0.15cm}
\label{tradeoff 2d}
\end{figure}

\subsubsection{Tradeoff between Consumption and Accuracy}
In Fig.~\ref{tradeoff 2d}, we illustrate the tradeoff between the battery consumption and estimation accuracy.
To better illustrate the tradeoff, we use $average$ $error$ as a merit, expressed as
\vspace{0.2cm}
\begin{equation}
\overline{ERR} = \frac{1}{n}\ \sum_{i=1}^{n} \bigg( \frac{\hat{C}_f(i)-C_f(i)}{C_f(i)} \bigg)^2,
\vspace{0.2cm}
\end{equation}
where $n$, $\hat{C}_f(i)$ and $C_f(i)$ are the same in (\ref{aea}).
We plot the curves of system's power consumption and average estimation error versus $\textit{PDT}$.

Fig.~\ref{tradeoff 2d} illustrates the relationship between the accuracy and the battery consumption.
Intuitively, a larger $\textit{PDT}$ introduces less power consumption, which proves that with a higher $\textit{PDT}$, consumption declines as the number of measured cubes decreases. Moreover, when $\textit{PDT} \ge 0.4$, the total consumption of the whole system can be reduced by $90\%$. The rapid decline of consumption is also related to the high redundancy of data in the typical 2D space as the roadside park.
On the other hand, the average error of ARMS increases as $\textit{PDT}$ becomes larger, which confirms the existence of the tradeoff between power consumption and estimation accuracy. Under this circumstance, choose $\textit{PDT}=0.41$ can achieve a relatively high predicting accuracy~(over $80\%$) while greatly reduce the battery consumption of the system.

\vspace{0.5cm}
% ==================================================Scenario II
\section{Application Scenario II: Performance Analysis in Vertical Enclosed Space}
In this section, we implement the adaptive monitoring algorithm in a typical 3D scenario, vertical enclosed space. We then present performance analysis of the GPM-NN and the adaptive monitoring algorithm in this typical scenario, respectively.

\begin{figure}[!t]
\small
\centering
\includegraphics[width=3.4in]{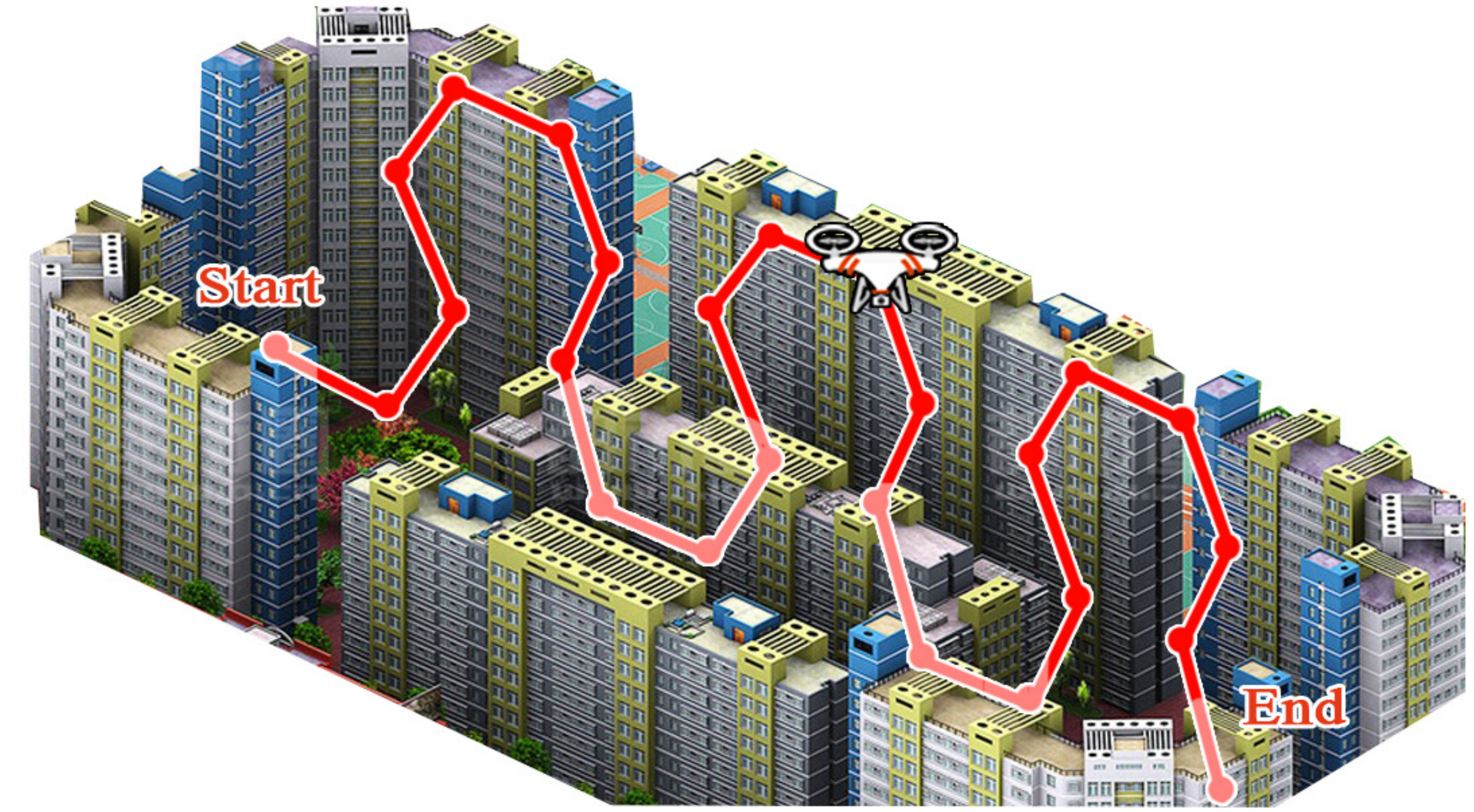}
\caption{The typical application scenarios of ARMS in 3D space~(courtyard inside a high-rise building).}
%\vspace{-0.15cm}
\label{system model2}
\end{figure}

\subsection{Scenario Description}
In the typical 3D scenario, the 3D space has target cubes in various heights. In this type of scenario, the planar area is relatively limited~(e.g., the courtyard inside a high-rise building).
As shown in Fig.~\ref{system model2}, in such a vertical enclosed space, there is no significant difference on AQI values between two horizontally neighboring cubes, but the wind may create a discrepancy of the pollutant concentration on two cubes at different heights.
Hence, the benefit of selecting more cubes vertically outweigh the cost of traversing between distant cubes at the same heights.

\subsection{Performance Analysis}
In this section, we present performance analysis of ARMS in different aspects, as in Section V.B, for typical 3D scenario.

\subsubsection{Model Accuracy}
In Fig.~\ref{accuracy 3d}, we compare three prediction models.
In the vertical enclosed space scenario, GPM-NN still maintains the highest accuracy among three models via different $\textit{PDT}$s. Compared to 2D scenario, LI decreases rapidly as $\textit{PDT}$ increases, which indicates the heterogenous in 3D AQI distribution. Moreover, when $\textit{PDT}=0.8$, GPM-NN would experience a violent decline. This phenomenon is caused by the inherent characteristic of $\textit{PDT}$. When $\textit{PDT}$ is high, the corresponding number of target cubes in $\mathcal M$ becomes so few that the predicting accuracy can significantly drop, even if only one point unmeasured~(e.g., 10 cubes with $\textit{PDT}=0.75$ and 9 cubes with $\textit{PDT}=0.8$).
This result can provide the basis for choosing the suitable $\textit{PDT}$ value.

In conclusion, GPM-NN performs better in both 2D and 3D fine-grained scenarios, with high estimation accuracy even if measuring cubes are few.

\begin{figure}[!t]
\small
\centering
\includegraphics[width=3.3in]{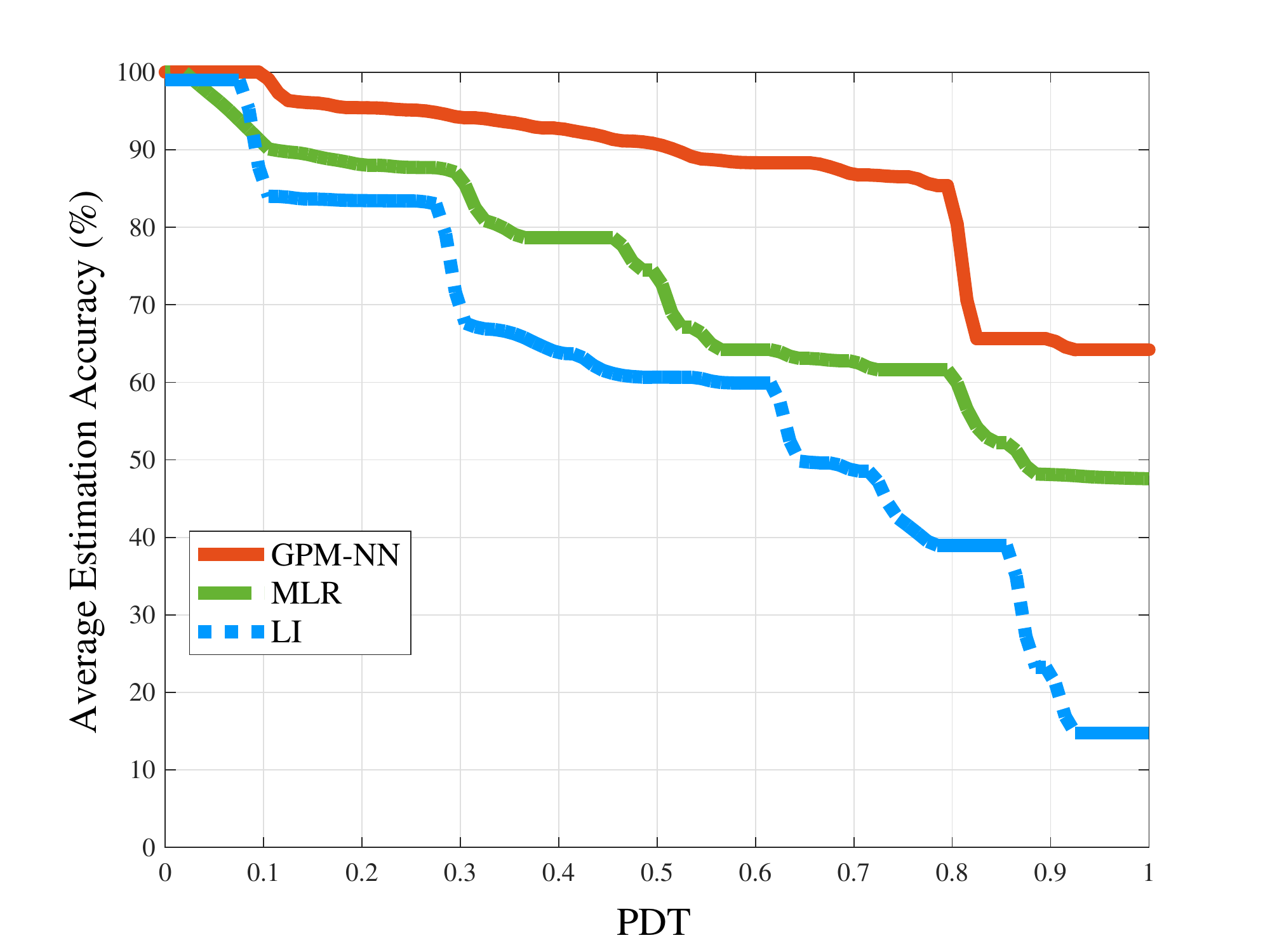}
\caption{The comparison of estimation accuracy between GPM-NN, MLR and LI, in 3D scenario.}
%\vspace{-0.15cm}
\label{accuracy 3d}
\end{figure}

\subsubsection{Effects of Neuron Numbers}
In Fig.~\ref{layernum 3d}, we study the effects of the number of neurons in a typical 3D scenario. When $\textit{PDT}<0.1$, the result is the same as in the 2D scenario, that each curve performs the best. As $\textit{PDT}$ increases, the curve with the number of neurons $=$ 0 declines most rapidly like that in Fig.~\ref{layernum 2d}. Also, the curve with fewer number of neurons~(e.g., the number of neurons $=$ 10) performs worse than with more neurons~(e.g., the number of neurons $=$ 100/1000) as well. In this scenario, we can find that the number of neurons $=$ 500 can achieve the highest estimation accuracy, which is different from the result in the 2D scenario. %This indicates that the chosen of number of neurons in hidden layer is relevant to its application scenario.

\begin{figure}[!t]
\small
\centering
\includegraphics[width=3.3in]{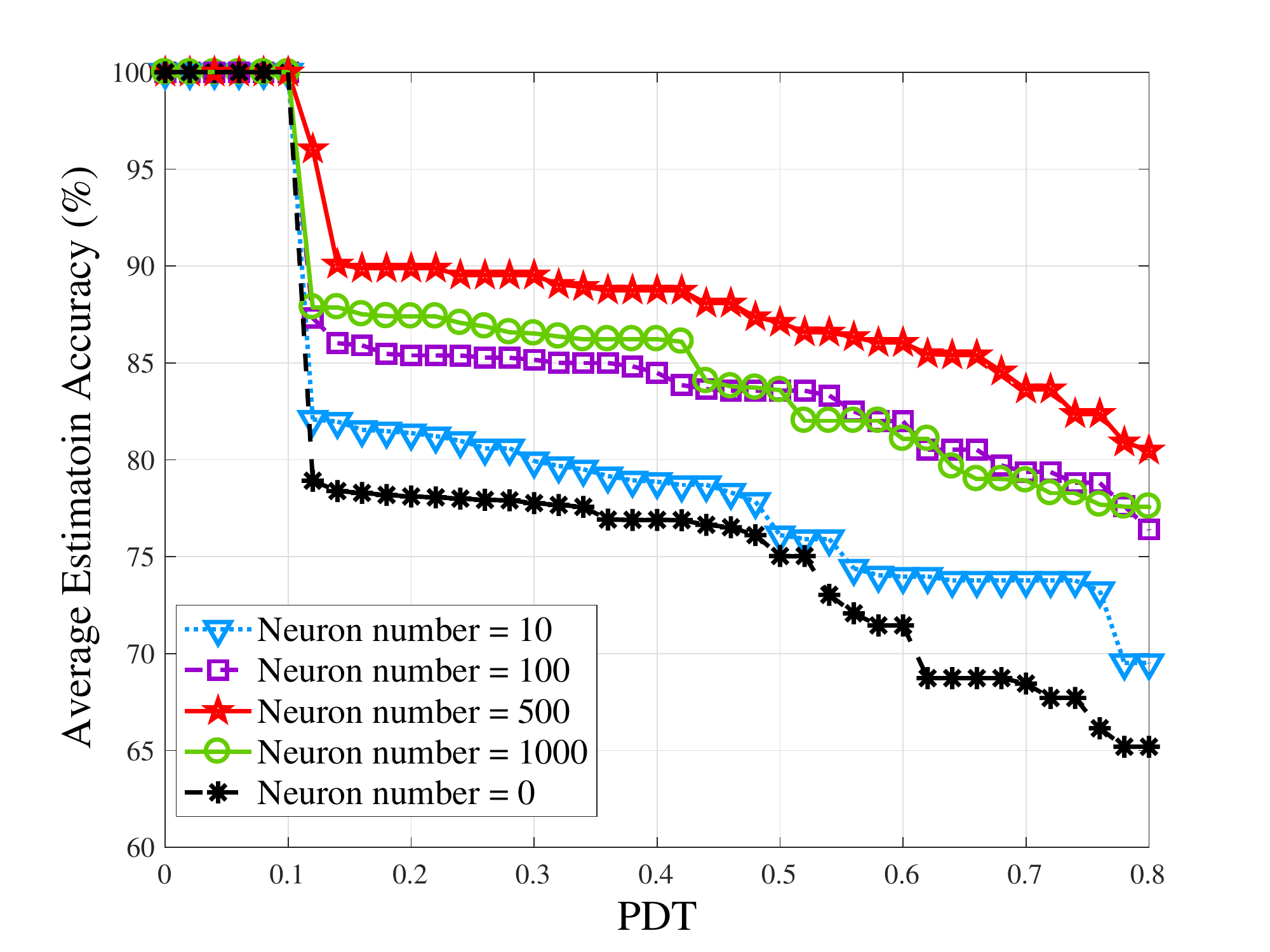}
\caption{The impact of the number of neurons in the non-linear part, in 3D scenario.}
%\vspace{-0.15cm}
\label{layernum 3d}
\end{figure}

In conclusion, our GPM-NN model~(with combination of linear and non-linear part) is robust and better than that with only linear part. Moreover, the number of neurons in the hidden layer can effectively influence the model's performance, and the optimal value is different in various scenarios.

\begin{figure}[!t]
\small
\centering
\includegraphics[width=3.3in]{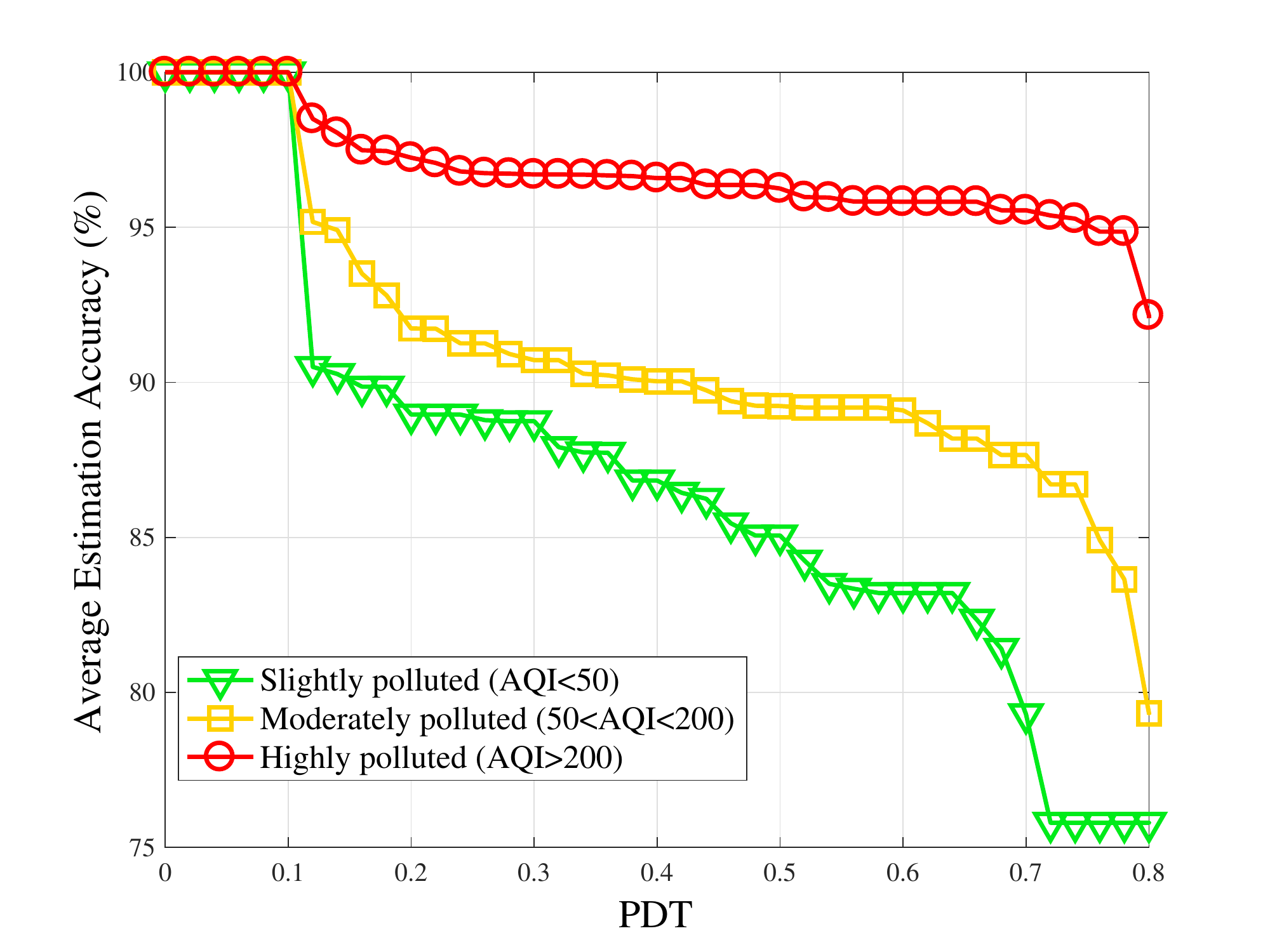}
\caption{The performance of GPM-NN with different AQI value, in 3D scenario.}
%\vspace{-0.15cm}
\label{different aqi 3d}
\end{figure}

\begin{figure}[!t]
\small
\centering
\includegraphics[width=3.3in]{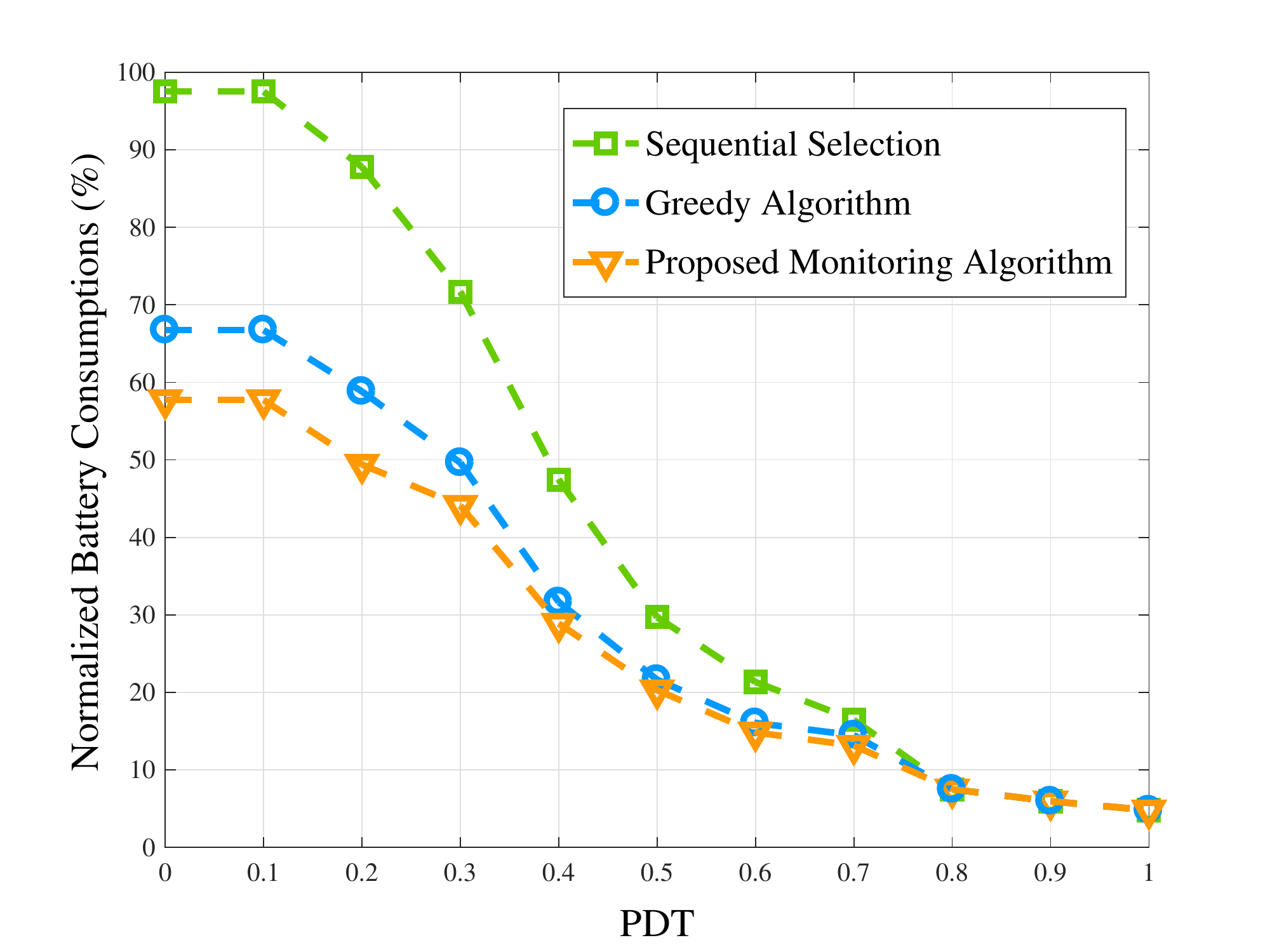}
\caption{Comparison of the Adaptive Monitoring Algorithm, Greedy Algorithm and Sequential Selection.}
%\vspace{-0.15cm}
\label{VGS}
\end{figure}

\subsubsection{Effects of Various AQI}
In Fig.~\ref{different aqi 3d}, we again plot the estimation accuracy of GPM-NN with different AQI values in the 3D sceanrio.
From the curves, we can find that GPM-NN also performs the best when moderately and highly polluted, while relatively worse when AQI is low.

In conclusion, GPM-NN can maintain better estimation accuracy when the AQI value is moderate and high, which is suitable for the operation of our ARMS.

\subsubsection{Performance of Adaptive Monitoring Algorithm}
In the 3D scenario as vertical enclosed space, Fig.~\ref{VGS} shows the consumption of three algorithms, our monitoring algorithm, greedy algorithm and sequential selection, via different $\textit{PDT}$s. From the figure, we can see when $\textit{PDT}$ is low, sequential selection consumes much more than those of our method and greedy algorithm. This indicates that when scenario becomes 3D, the cube selection can be more complicated and a suitable selection method can highly reduce the battery consumption. Moreover, adaptive monitoring algorithm also performs the best among three methods, and it is better than the greedy algorithm when $PDT\leq 0.8$. As $\textit{PDT}$ becomes high, the normalized consumption of three algorithms is closer, and becomes equal when $PDT\geq 0.8$. Thus, the adaptive monitoring algorithm can effectively save the battery life for monitoring AQI in 3D scenario.

\subsubsection{Tradeoff between Consumption and Accuracy}
In Fig.~\ref{tradeoff 3d}, we plot the tradeoff in the 3D scenario as horizontal enclosed space. This typical 3D scenario is more common in real measurement, and hence the result is more instructive.
As $\textit{PDT}$ becomes higher, the average error grows rapidly as consumption can drop fairly. Given the average error, for example, when $\overline{ERR}=0.04$~(average estimation accuracy is about $80\%$), the corresponding $\textit{PDT}=0.51$, and thus the power consumption can be reduced to as little as $37\%$. Hence, by choosing suitable $\textit{PDT}$ value for monitoring, the measuring efforts can greatly scale down.
%In conclusion, Arms can well balance the inherent tradeoff, by highly reducing the system overhead as well as maintain the sufficiently-high accuracy of a 3D AQI map.

\begin{figure}[!t]
\small
\centering
\includegraphics[width=3.3in]{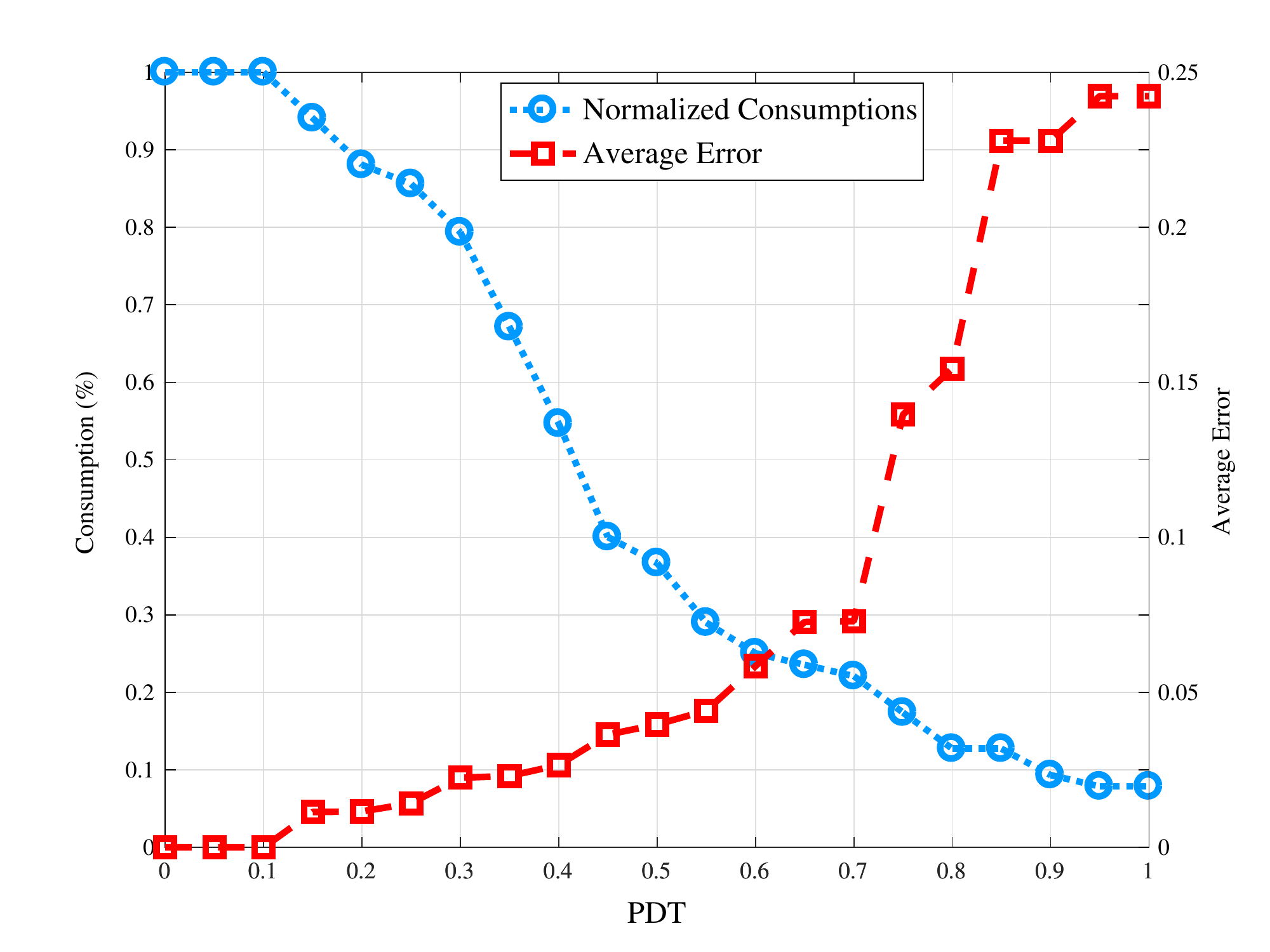}
\caption{The tradeoff between system battery consumption and estimation accuracy, in 3D scenario.}
%\vspace{-0.15cm}
\label{tradeoff 3d}
\end{figure}

%\vspace{0.2cm}
% ==================================================Conclusion
\section{Conclusion}
In this paper, we have designed a UAV sensing system, ARMS, to construct fine-grained AQI maps.
% adaptive monitoring techniques for mobile sensing system are addressed for efficiency and effectiveness
A novel fine-grained AQI distribution model GPM-NN has been proposed based on NN and physical model, to help generate a realtime AQI map with data collected by ARMS.
To reduce the battery consumptions of ARMS, we have proposed the adaptive monitoring algorithm to efficiently update realtime AQI maps. For the 2D and 3D scenarios, we have applied the adaptive monitoring algorithm, respectively. By using the proposed index $\textit{PDT}$, the system can well balance the intrinsic tradeoff between the estimation accuracy and power consumption.
Experimental results have showed that GPM-NN can achieve a higher accuracy in AQI map construction than other existing models, and the number of neurons in the hidden layer of GPM-NN should also be adjusted in various scenarios to acquire better performance. Moreover, the adaptive monitoring algorithm can generate trajectory while greatly saving the battery life of the UAV, and ARMS can well balance the tradeoff between accuracy of AQI map and battery consumptions.

%\vspace{0.2cm}
% ==================================================Appendix: Proof of Minimum Residual Existence
\section*{Appendix A\\ Proof of Proposition 1}
For $\beta_j$ where $j \in [1,K$+$2]$, we have
%\vspace{0.2cm}
\begin{equation}
\frac{\partial^{2}\mathcal{S}}{\partial \beta_j^{2}} =
\left\{
\begin{array}{ll}
2 \sum_{i=1}^{N} g_j^2 > 0, & 1 \leq j \leq K,\\[12pt]
2 \sum_{i=1}^{N} C^2(\vec{x}_i,u_i) > 0, \quad & j=K$+1$,\\[12pt]
2 \sum_{i=1}^{N} 1 = 2N > 0, & j=K$+2$.
\end{array}
\right.
%\vspace{0.2cm}
\end{equation}
Hence, $\partial\mathcal{S}/\partial \beta_j$ are all convex functions, with $j \in [1,K$+$2]$.

As for variable $H$, the second order partial derivative can be calculated as
\vspace{0.28cm}
\begin{small}
\begin{equation*}
\begin{split}
\frac{\partial^{2}\mathcal{S}}{\partial H^{2}} & =\\
& -2\sum_{i=1}^{N} \beta^{'} \Bigg[ -\frac{\beta^{'}}{\sigma_z^6u_i^2}(z_i-H)^2 \exp\left(-\frac{(z_i-H)^2}{\sigma_z^2}\right) \ - \\
& \left(C - \frac{\beta^{'} \exp\left(-\frac{(z_i-H)^2}{2\sigma_z^2}\right)}{u_i \sigma_z}\right)\frac{1}{u_i \sigma_z^3} \exp\left(-\frac{(z_i-H)^2}{2\sigma_z^2}\right) \ + \\
& \left(C - \frac{\beta^{'} \exp\left(-\frac{(z_i-H)^2}{2\sigma_z^2}\right)}{u_i \sigma_z}\right) \frac{(z_i-H)^2 \exp\left(-\frac{(z_i-H)^2}{2\sigma_z^2}\right)}{u_i \sigma_z^5} \Bigg],
\vspace{0.3cm}
\end{split}
\end{equation*}
\end{small}
where $C = \left(\hat{C}_f(\vec{x}_i,u_i)- C_{static} -\beta_{K+2} - \sum_{j=1}^{K}\beta_j g_j\right)$, and $\beta^{'}=\frac{\lambda}{\sqrt{2\pi}}\beta_{K+1}\left(1-2Q\left(\frac{L}{2\sigma_y}\right)\right)$. Then we have
\vspace{0.2cm}
\begin{small}
\begin{equation*}
\begin{split}
\frac{\partial^{2}\mathcal{S}}{\partial H^{2}} & = \\
& 2\sum_{i=1}^{N} \Bigg[ \left(\frac{2\beta^{'2}(z_i-H)^2}{\sigma_z^6u_i^2} - \frac{\beta^{'2}}{\sigma_z^4u_i^2}\right) \exp\left(-\frac{(z_i-H)^2}{\sigma_z^2}\right) \ + \\
& C \left(\frac{\beta^{'}}{\sigma_z^3u_i} - \frac{\beta^{'} (z_i-H)^2}{u_i\sigma_z^5}\right) \exp\left(-\frac{(z_i-H)^2}{2\sigma_z^2}\right) \Bigg].
\vspace{0.3cm}
\end{split}
\end{equation*}
\end{small}

Let $t_i = \exp\left(-\frac{(z_i-H)^2}{2\sigma_z^2}\right)$, each item of the summation is equivalent to a quadratic function
$Q_i(t_i) = a_i t_i^2 + b_i t_i$. Note that $t_i \in (0,1]$, and $t_i=0$ is one zero point of $Q_i(t_i)$. To satisfy the proposition that $\partial^{2}\mathcal{S}/\partial H^{2}$ always has positive value, the problem becomes
\vspace{0.3cm}
\begin{equation*}
\left\{
\begin{aligned}
\begin{split}
& a_i = \frac{2\beta^{'2}(z_i-H)^2}{u_i^2 \sigma_z^6}-\frac{\beta^{'2}}{u_i^2 \sigma_z^4}  < 0, \\
& b_i = C \left( \frac{\beta^{'}}{u_i \sigma_z^3}-\frac{\beta^{'}(z_i-H)^2}{u_i \sigma_z^5} \right) > 0,
\end{split}
\end{aligned}
\right.
\quad \  \forall i \in [1,N],
\vspace{0.3cm}
\end{equation*}
which can be simplified as:
\vspace{0.3cm}
\begin{equation}
\left\{
\begin{aligned}
\begin{split}
& \sigma_z^2 > \max \limits_{i} \ 2(z_i-H)^2, \\
& \sigma_z^2 > \max \limits_{i} \ (z_i-H)^2.
\end{split}
\end{aligned}
\right.
\label{condition}
\vspace{0.3cm}
\end{equation}

We define $H \in [0,H_0]$, where $H_0$ is the upper bound for a fine-grained measurement.
%Thus, equation (\ref{condition}) has a finite maximum value as $\max\{2z_i^2,2H_0^2\}$. Moreover, another zero point $x^{*} = -\frac{b}{a}>1$, on condition of a large $\sigma_z$.
Hence, by choosing appropriate diffusion parameter $\sigma_z$ as $\sigma_z^2 > \max\{2z_i^2,2H_0^2\}$, we have
\vspace{0.3cm}
\begin{equation*}
\begin{split}
\frac{\partial^{2}\mathcal{S}}{\partial H^{2}} & = 2\sum_{i=1}^{N} Q_i(t_i) = 2\sum_{i=1}^{N} (a_i t_i^2 + b_i t_i) \\
& = 2\sum_{i=1}^{N} \left( \frac{b_i^2}{4|a_i|} - |a_i|\left(t_i + \frac{b_i}{2a_i}\right)^2 \right) > 0, \ \forall t_i \in (0,1].
\end{split}
\end{equation*}
\vspace{0.25cm}

Therefore, $\partial\mathcal{S}/\partial H$ is also a convex function, which indicates that equation (\ref{residual}) has a minimum as well as a unique value, correspondingly.

\iffalse
% ==================================================Acknowledgement
\section*{Acknowledgment}
The authors would like to thank...
\fi

\ifCLASSOPTIONcaptionsoff
  \newpage
\fi

% ==================================================Reference

\iffalse
% ==================================================Biography
\begin{IEEEbiography}{Yuzhe Yang}
Biography text here.
\end{IEEEbiography}

\begin{IEEEbiography}{Zijie Zheng}
Biography text here.
\end{IEEEbiography}

\begin{IEEEbiography}{Lingyang Song}
Biography text here.
\end{IEEEbiography}
\fi

% ==================================================that's all folks

\begin{thebibliography}{1}

\bibitem{who}
W. H. Organization, ``7 million premature deaths annually linked to air pollution,'' \emph{Air Quality \& Climate Change},
vol. 22, no. 1, pp. 53-59, Mar. 2014.

\bibitem{air pollution mortality}
Q. Di, Y. Wang, A. Zanobetti, et al, ``Air pollution and mortality in the medicare population,'' \emph{New England J. of Medicine}, vol. 376, no. 26, pp. 2513-2522. Jul. 2017.

\bibitem{airmap}
Y. Li, Y. Zhu, W. Yin, Y. Liu, G. Shi and Z. Han, ``Prediction of High Resolution Spatial-Temporal Air Pollutant Map from Big Data Sources,'' \emph{Int. Conference on Big Data Computing and Commun.}, Taiyuan, China, pp. 273-282. Jul. 2015.

\bibitem{station}
B. Zou, J. G. Wilson, F. B. Zhan, and Y. N. Zeng, ``Air pollution exposure assessment methods utilized in epidemiological studies,'' \emph{J. of Environmental Monitoring}, vol. 11, no. 3, pp. 475-490, Feb. 2009.

\bibitem{monitor station}
Beijing MEMC, ``Beijing municipal environmental monitoring center,'' \emph{http://www.bjmemc.com.cn/.} Mar. 2017.

\bibitem{mobisys}
Y. Cheng, X. Li, Z. Li, S. Jiang, Y. Li, J. Jia, and X. Jiang, ``Aircloud: a cloud-based air-quality
monitoring system for everyone,'' \emph{Proc. of the 12th ACM Conference on Embedded Network Sensor Syst.}, New York, NY, Nov. 2014.

\bibitem{node1}
D. Hasenfratz, O. Saukh, C. Walser, C. Hueglin, M. Fierz, T. Arn, J. Beutel, and L. Thiele, ``Deriving high-resolution urban air pollution maps using mobile sensor nodes,'' \emph{Pervasive and Mobile Compting}, vol. 16, no. 2, pp. 268-285, Jan. 2015.

\bibitem{node2}
N. Nikzad, N. Verma, C. Ziftci, E. Bales, N. Quick, P. Zappi, K. Patrick, S. Dasgupta, I. Krueger, T. Rosing, and W. Griswold, ``CitiSense: improving geospatial environmental assessment of air quality using a wireless personal exposure monitoring system,'' \emph{Proc. of ACM Wireless Health}, San Diego, CA, Oct. 2010.

\bibitem{node3}
Y. Gao, W. Dong, K. Guo, X. Liu, Y. Chen, X. Liu, J. Bu and C. Chen, ``Mosaic: a low-cost mobile sensing system for urban air quality monitoring,'' \emph{IEEE Int. Conference on Comput. Commun.~(INFOCOM'16)}, San Francisco, CA, Jul. 2016.

\bibitem{balloon}
D. Bisht, S. Tiwari, U. Dumka, A. Srivastava, P. Safai, S. Ghude, D. Chate, P. Rao, K. Ali, T. Prabhakaran, et al, ``Tethered balloon-born and ground-based measurements of black carbon and particulate profiles within the lower troposphere during the foggy period in delhi, India,'' \emph{Sci. of The Total Environment}, vol. 573, no. 1, pp. 894-905. Dec. 2016.

\bibitem{blueaer}
Y. Hu, G. Dai, J. Fan, Y. Wu and H. Zhang, ``BlueAer: A fine-grained urban PM2.5 3D monitoring system using mobile sensing,'' \emph{IEEE Int. Conference on Comput. Commun.~(INFOCOM'16)}, San Francisco, CA, Jul. 2016.

\bibitem{vertical dis}
T. N. Quang, C. He, L. Morawska, L. D. Knibbs, and M. Falk, ``Vertical particle concentration profiles around urban office buildings,'' \emph{Atmospheric Chemistry and Physics}, vol. 12, no. 11, pp. 5017-5030. May 2012.

\bibitem{height pro}
F. M. Rubinoa, L. Floridiaa, M. Tavazzania, S. Fustinonia, R. Giampiccoloa, A. Colombia, ``Height profile of some air quality markers in the urban atmosphere surrounding 100m tower building,'' \emph{Atmospheric Environment}, vol. 32, no. 20, pp. 3569-3580. Sep. 1998.

\bibitem{whyfineaqi}
C. Borrego, H. Martins, O. Tchepel, L. Salmim, A. Monteiro, and A. I. Miranda, ``How urban structure can affect city sustainability from an air quality perspective,'' \emph{Environmental modelling \& software}, vol. 21, no. 4, pp. 461-467, Apr. 2006.

\bibitem{uair}
Y. Zheng and F. Liu and H. Hsieh,
``U-Air: when urban air quality inference meets big data,'' \emph{Proc. of the 19th ACM SIGKDD int. conference on Knowledge discovery and data mining~(KDD '13)}, Chicago, IL, pp. 1436-1444, Aug. 2013.

\bibitem{model2}
H. X. Xu, G. Li, S. L. Yang and X. Xu, ``Modeling and simulation of haze process based on Gaussian model,'' \emph{2014 11th Int. Comput. Conference on Wavelet Actiev Media Technology and Inform. Process.(ICCWAMTIP)}, Chengdu, China, pp. 68-74, Apr. 2014.

\bibitem{model3}
Michela Cameletti, Rosaria Ignaccolo and Stefano Bande, ``Comparing spatio-temporal models for particulate matter in Piemonte,'' \emph{Environmetrics}, vol. 22, no. 8, pp. 985-996. Dec. 2011.

\bibitem{NN1}
C. Zhao, M. Heeswijk and J. Karhunen, ``Air quality forecasting using neural networks,'' \emph{IEEE Symp. Series on Computational Intell.~(SSCI)}, Athens, Greece, Dec. 2016.

\bibitem{NN2}
M. Cai, Y. Yin and M. Xie, ``Prediction of hourly air pollutant concentrations near urban arterials using artificial neural network approach,'' \emph{Transportation Research Part D: Transport and Environment}, vol. 14, no. 1, pp. 32-41, Jan. 2009.

\bibitem{NN3}
M. W. Gardner, S. R. Dorling, ``Neural network modelling and prediction of hourly NOx and NO2 concentrations in urban air in London,'' \emph{Atmospheric Environment}, vol. 33, no. 5, pp. 709-719, Feb. 1999.

\bibitem{NN4}
M. Dedovic, S. Avdakovic, I. Turkovic, N. Dautbasic, and T. Konjic, ``Forecasting PM10 concentrations using neural networks and system for improving air quality,'' \emph{2016 XI Int. Symp. on Telecommun.~(BIHTEL)}, Sarajevo, Bosnia-Herzegovina, Oct. 2016.

\bibitem{chinaweather}
Beijing EPB, ``Beijing municipal environmental protection bureau,'' \emph{http://www.bjepb.gov.cn/.} Mar. 2017.

\bibitem{sensor}
Plantower, ``Technology laser PM2.5 sensor, air quality sensor,'' \emph{http://www.plantower.com/en/.}

\bibitem{uav}
Da-Jiang Innovations Science and Technology Co., Ltd.~(DJI), Phantom 3 Professional. \emph{https://www.dji.com/cn/phantom-3-pro.}

\bibitem{statistics book}
R. V. Hogg and A. T. Craig, \emph{Introduction to mathematical statistics}, 5th ed. Upper Saddle River, New Jersey: Prentice Hall. 1995.

\bibitem{sptialmodel}
M. Cameletti, F. Lindgren, D. Simpson, and H. Rue, ``Spatio-temporal modeling of particulate matter concentration through the SPDE approach,'' \emph{Advances in Statistical Anal.}, vol. 97, no. 2, pp. 109-131. Apr. 2013.

\bibitem{model1}
D. R. Middleton, ``Modelling air pollution transport and deposition,'' \emph{IEE Colloquium on Pollution of Land, Sea and Air: An Overview for Engineers}, London, UK, Oct. 1995.

\iffalse
\bibitem{dataset1}
World-wide Air Quality Monitoring Data Coverage. \emph{https://aqicn.org/sources/.}

\bibitem{dataset2}
U.S. Environmental Protection Agency. ''Air Data: Air Quality Data Collected at Outdoor Monitors Across the US,'' \emph{https://www.epa.gov/outdoor-air-quality-data.}

\bibitem{dataset3}
Mainland China Key City Air Quality~(AQI) Historical Database. \emph{https://www.gracecode.com/aqi.html.}
\fi

\bibitem{mathmodel}
J. M. Stockie, ``The mathematics of atmospheric dispersion modeling,'' \emph{Siam Review}, vol. 53, no. 2, pp. 349-372. May 2011.

\bibitem{robust}
S. Brusca, F. Famoso, R. Lanzafame, S. Mauro, A. Marino Cugno Garrano, and P. Monforte, ``Theoretical and experimental study of Gaussian plume model in small scale system,'' \emph{Energy Procedia}, vol. 101, no. 1, pp. 58-65. Nov. 2016.

\bibitem{Revision:CNN}
F. Tang, B. Mao, Z. Fadlullah, N. Kato, O. Akashi, T. Inoue, and K. Mizutani, ``On Removing Routing Protocol from Future Wireless Networks: A Real-time Deep Learning Approach for Intelligent Traffic Control,'' \emph{IEEE Wirelesss Mag.~(WCM)}, In press. 2017.

\bibitem{Revision:RNN}
A. Al-Molegi, M. Jabreel, and B. Ghaleb, ``STF-RNN: Space Time Features-based Recurrent Neural Network for predicting people next location,'' \emph{2016 IEEE Symp. Series on Computational Intell.~(SSCI)}, Athens, Greece, Dec. 2016.

\bibitem{nnproof1}
S. M. Carroll and B. W. Dickinson, ``Construction of neural nets using the Radon transform,'' \emph{Proc. IEEE 1989 Int. Joint Conf. on Neural Networks}, New York, Feb. 1989.

\bibitem{nnproof2}
G. Cybenko, ``Approximation by superpositions of a sigmoidal function,'' \emph{Math. Control, Signals, and Syst.}, vol. 2, no. 4, pp. 303-314. Feb. 1989.

\bibitem{nnproof3}
K. Funahashi, ``On the approximate realization of continuous mapping by neural networks,'' \emph{Neural Networks}, vol. 2, no. 3, pp. 183-192, Feb. 1989.

\bibitem{newton method}
D. P. Bertsekas, \emph{Nonlinear programming}, Belmont: Athena scientific. 1999. pp. 1-60.

\bibitem{inverse matrix}
R. Penrose, ``A generalized inverse for matrices,'' \emph{Math. proc. of the Cambridge philosophical soc.}, vol. 51, no. 3, pp. 406-413. Jul. 1955.

\bibitem{moore penrose}
R. MacAusland, ``The moore-penrose inverse and least squares,'' \emph{Math 420: Advanced Topics in Linear Algebra}. 2014.

\bibitem{TSP}
D. Goldberg, R. Lingle, ``Alleles, loci, and the traveling salesman problem,'' \emph{Proc. of an Int. Conf. on Genetic Algorithms and Their Applicat.}, vol. 154, pp. 154-159. Hillsdale, NJ, Jul. 1985.

\bibitem{LI}
C. De Boor, \emph{A practical guide to splines}. New York: Springer-Verlag. 1978.

\iffalse
\bibitem{TSP1}
C. Nilsson, ``Heuristics for the traveling salesman problem,'' 2003.

\bibitem{TSP2}
T. Yamada, K. Aihara, and M. Kotani, ``Chaotic neural networks and the traveling salesman problem,'' \emph{Proc. of 1993 Int. Joint Conf. on Neural Networks}, Nagoya, Japan. vol. 2, pp. 1549-1552. Oct. 1993.
\fi

\end{thebibliography}
\end{document}